\documentclass[10pt,journal]{IEEEtran}
\usepackage{amssymb}
\usepackage{amsmath}
\usepackage{cite}
\usepackage{url}
\usepackage{xcolor}
\usepackage{cite,graphicx,amsmath,amssymb}
\usepackage{fancyhdr}
\usepackage{mdwmath}
\usepackage{mdwtab}
\usepackage{caption}
\usepackage{amsthm}
\usepackage{setspace}
\usepackage{hyperref}
\usepackage{algorithm}
\usepackage{algorithmic}
\usepackage{multirow}
\usepackage{makecell}
\usepackage{mathtools}
\usepackage{subcaption}
\usepackage{bm}
\usepackage{tikz}
\usepackage{xurl}
\usepackage{stfloats}
\usepackage{mathrsfs}

\usepackage{booktabs}
\usepackage{multirow}
\usepackage{siunitx}
\usepackage{balance}

\hypersetup{colorlinks=true,
linkcolor=blue,
citecolor=blue,      
urlcolor=black,
}

\graphicspath{{./src/img/}}

\DeclareMathOperator*{\argmax}{arg\,max}

\newtheorem{remark}{Remark}
\newtheorem{theorem}{Theorem}

\newtheorem{lemma}{Lemma}

\newtheorem{assumption}{Assumption}

\newtheorem{corollary}{Corollary}

\newtheorem{proposition}{Proposition}

\allowdisplaybreaks
\NewDocumentCommand{\multiubrace}{mmm}
 {
  \egreg_multiubrace:nnn {#1} {#2} {#3}
 }

\expandafter\def\expandafter\normalsize\expandafter{%
	\normalsize%
	\setlength\abovedisplayskip{4pt}%
	\setlength\belowdisplayskip{4pt}%
	\setlength\abovedisplayshortskip{2pt}%
	\setlength\belowdisplayshortskip{2pt}%
}
\begin{document}
\title{SiFo: Wireless Foundation Model for Low-Overhead Site-Specific CSI Feedback}
\author{{Cheng-Jie Zhao~\IEEEmembership{Graduate Student Member, IEEE}, Zhaolin Wang,~\IEEEmembership{Member, IEEE}, \\ Zongyao Zhao,~\IEEEmembership{Member, IEEE}, and Yuanwei Liu,~\IEEEmembership{Fellow, IEEE}}
	\vspace{-0.4cm}
\thanks{The authors are with Department of Electrical and Computer Engineering, The University of Hong Kong, Hong Kong (e-mail: chengjie\_zhao@connect.hku.hk, zhaolin.wang@hku.hk, zongyao@hku.hk, yuanwei@hku.hk).}
}
\maketitle

\begin{abstract}
SiFo, a wireless foundation model-based framework, is proposed for low-overhead site-specific channel state information (CSI) feedback. In 3GPP NR, Type-II feedback provides an expressive codebook-based CSI representation, but it requires substantial reference-signal overhead, UE-side search, and feedback. Learning-based site-specific feedback can reduce these online costs while retaining high-quality subspace representation by exploiting deployment-dependent propagation structure. However, existing site-specific designs typically train a dedicated neural network for each new site, which limits scalability when the number of deployments is large. SiFo addresses this scalability issue by pretraining a CSI feedback model across source sites and adapting it to a target site through lightweight calibration. A small set of target-site users reports low-dimensional reference signal received power (RSRP) fingerprints, and their full-CSI-based subspace labels are stored as calibration memory. During online operation, a served user is matched to calibrated users through the same SSB probing and RSRP reporting procedure, so nearby calibration samples provide site-specific subspace guidance without updating model parameters. SiFo therefore transfers common propagation knowledge while retaining local adaptation. Numerical results across ten city scenarios demonstrate that SiFo (i) achieves higher CSI-capture efficiency than separately trained site-specific learning baselines under the same target-site labeled budget, (ii) approaches the high-overhead 3GPP NR Type-II feedback reference using only RSRP measurements collected during online SSB probing, and (iii) converts the high CSI-capture efficiency and low overhead into effective spectral efficiency improvement under limited target-site data.
\end{abstract}


\section{Introduction}
\IEEEPARstart{W}{ireless} foundation models are being developed as general-purpose backbones for learning-enabled wireless systems. This direction is driven by the increasing use of learning-based modules in a wide range of wireless applications, e.g., channel acquisition, beam management, and network orchestration. In practical deployments, these modules are expected to operate across heterogeneous carrier bands, antenna arrays, protocol configurations, and propagation environments. Such heterogeneity exposes a scalability limitation of independently trained models, because labeled channel-data collection, model retraining, and deployment-specific validation may need to be repeated whenever the radio environment changes. Wireless foundation models therefore shift the design emphasis from isolated task- or site-specific learners toward reusable wireless representations that can be specialized to new deployments with limited additional data.
\\\indent
In the broad sense of \cite{Bommasani2021Foundation}, a foundation model is pretrained on large and diverse data and can later be adapted efficiently to downstream environments or tasks. Existing wireless foundation model studies have explored this principle first from network intelligence and channel modeling perspectives. The authors of \cite{WirelessLLM2024} proposed WirelessLLM, which exploits large language models for wireless intelligence and network orchestration. A space-time-frequency wireless foundation model, named WiFo, was developed in \cite{WiFo2024} for transferable channel prediction. The large wireless model in \cite{Alikhani2024LWM} learned contextualized channel embeddings for downstream communication and sensing tasks. For multi-task wireless communication and sensing, the authors of \cite{Yang2025WirelessGPT} introduced WirelessGPT through generative pretraining, and \cite{Sheng2025WirelessFM} further extended wireless foundation models to multi-task prediction.
\\\indent
Another line of work moves foundation modeling closer to physical-layer signal representations and protocol-relevant tasks. To support sensing, communication, and localization, \cite{Aboulfotouh2025WavesFM} proposed WavesFM as a multi-modal wireless foundation model. From the spectrum-management perspective, SpectrumFM was developed in \cite{Zhou2026SpectrumFM} to learn transferable spectrum representations. At the signal-representation level, the authors of \cite{Mashaal2026IQFM} introduced IQFM, which learns reusable encoders directly from raw multi-antenna IQ streams. As a further step toward CSI feedback, \cite{Liu2026WiFoCF} proposed WiFo-CF, a foundation model tailored to heterogeneous channel dimensions, feedback rates, and data distributions. These studies demonstrate the potential of pretraining and transfer across network-level intelligence, channel prediction, and signal representation, spectrum management. How to use a wireless foundation model for low-overhead, protocol-constrained, and limited-feedback CSI acquisition remains less explored.
\\\indent
Low-overhead limited-feedback CSI acquisition has traditionally been handled through structured codebook-based schemes \cite{Love2008LimitedFeedback,Fu2023TutorialCodebooks}. In the 3rd Generation Partnership Project (3GPP) New Radio (NR) framework, Type-I, Type-II, and enhanced Type-II port-selection codebook (PSC) schemes are representative limited-feedback options \cite{3GPP38214}. Type-I has low feedback overhead but limited single-beam representation capability in rich multipath channels. Type-II provides the most expressive multi-beam representation over an oversampled angular dictionary, but its gain comes with heavier RS acquisition, user equipment (UE)-side search, and feedback processing. PSC provides a structured alternative, yet its representation capability remains constrained by the selected port domain. Together, these schemes expose the fundamental trade-off between affordable overhead and high-quality channel representation. 
\\\indent
Learning-based methods have been explored to improve the tradeoff among overhead, complexity, and performance. Early studies enhanced conventional beam-management and feedback procedures through learned codebooks. The authors of \cite{Dreifuerst2024MLCodebook} developed machine-learning-based codebooks for NR initial access and CSI Type-II feedback, and later extended this idea to neural codebook design for broader multiple-input multiple-output (MIMO) beam-management tasks \cite{Dreifuerst2025NeuralCodebook}. More recent works moved beyond generic codebook learning by exploiting site-specific propagation structure, including site-specific probing, RSRP-based codebook design, and channel knowledge maps \cite{Heng2022SiteSpecificProbing,Ning2023RSRPCodebook,Wu2024CKMBeamforming,Zeng2024CKMTutorial}. In particular, site-specific beamforming (SSBF) frameworks showed that coarse power fingerprints can be used to synthesize site-conditioned beams \cite{Heng2024GridFree, SSBFMAG2,sim}, while site-specific Type-II (SST2) feedback incorporated \emph{site-specific information (SSI)} into the conventional Type-II (Conv-T2) feedback pipeline to reduce online overhead and UE complexity without sacrificing the expressive Type-II representation \cite{Zhao2026UnifiedLF}. Together, these studies demonstrate the value of learning from the deployment environment. However, when such methods are extended to large-scale network deployment, their site-dependent training process becomes difficult to scale: a dedicated neural model may need to be trained and validated for each new site, and modern wireless networks may contain a very large number of deployment sites.
\\\indent
This scalability issue motivates a closer connection between wireless foundation models and site-specific learning. At first glance, the two directions appear to pursue opposite design philosophies: wireless foundation models emphasize reusable structures shared across tasks and deployments, whereas site-specific learning emphasizes adaptation to localized radio propagation for site-specific performance. These two goals can be made compatible if the feedback rule is decomposed into a reusable cross-site component and a lightweight site-conditioned component. Cross-site pretraining can learn common beam-domain sensing and subspace priors, while target-site calibration can supply the local propagation evidence that is not transferable across sites. The desired design is therefore neither a purely site-agnostic model nor a separately trained model for every site.
\\\indent
This paper studies foundation model-based site-specific CSI feedback in the Type-II setting. Type-II feedback is chosen because it provides the most expressive NR codebook-based CSI representation, and hence exposes the key difficulty of inferring a phase-aware CSI subspace from low-dimensional RSRP observations. The RSRP fingerprint contains only coarse amplitude information, whereas the desired subspace depends on multipath geometry, angular spread, phase, and scattering structure. Therefore, similar RSRP observations may correspond to different subspace decisions in different deployment sites. Direct cross-site pretraining provides a useful marginal prior, but it does not by itself recover the target-site relation from RSRP fingerprints to CSI subspaces.
\\\indent
Based on these observations, this paper proposes SiFo, a calibration-aided wireless foundation model design for low-overhead site-specific CSI feedback. The pretrained model provides a shared SSB probing codebook and a parametric subspace prior, while the target site is represented by a small calibration memory. Each calibrated UE is associated with an RSRP fingerprint and an offline full-CSI-based projector. During online operation, nearby calibrated UEs in the pretrained RSRP space provide local subspace evidence for the served UE, enabling target-site specialization without updating model parameters. The main contributions are summarized as follows:
\begin{itemize}
	\item We establish a foundation-model view of site-specific CSI feedback. In the Type-II setting, this formulation identifies the uncertainty in inferring CSI subspaces from RSRP fingerprints as the main obstacle to direct cross-site transfer and motivates SiFo as a combination of cross-site pretraining, target-site calibration, and online subspace retrieval.
	\item We develop a target-site calibration principle for low-overhead CSI acquisition. The analysis shows that RSRP sensing during SSB probing should provide discriminative directional-power fingerprints, while offline full-CSI-based projector labels supply the phase-aware subspace information missing from RSRP measurements.
	\item We design a projector-based acquisition rule for gradient-free target-site specialization. Online UEs are matched to calibrated UEs in the pretrained RSRP coordinates, and retrieved projectors are averaged and fused with the parametric predictor to produce the prescribed rank-$Q$ CSI representation subspace.
	\item We demonstrate the data and overhead efficiency of SiFo through ten-city simulations. Without target-site fine-tuning, SiFo outperforms single-site SST2 baselines under identical target-site data budgets, approaches the high-overhead Conv-T2 reference with much lower online interaction, and converts CSI-capture gains into effective-rate improvement.
\end{itemize}

The rest of this paper is organized as follows. Section~\ref{sec:sys_model} presents the system model. Section~\ref{sec:from_feedback_to_fm_sst2} formulates SiFo. Sections~\ref{sec:sensing_criteria} and~\ref{sec:site_sensing_encoding} develop calibration and projector-memory acquisition. Section~\ref{sec:backbone_deployment} gives the deployment protocol. Section~\ref{sec:results} reports simulations, and Section~\ref{sec:conclusion} concludes the paper.

\indent \textit{Notation:} Scalars, vectors, matrices, and sets are denoted by italic, boldface lowercase, boldface uppercase, and calligraphic letters, respectively. $\mathbb C$ and $\mathbb R$ denote the complex and real fields. $(\cdot)^T$ and $(\cdot)^H$ denote transpose and Hermitian transpose. For vectors, $\|\cdot\|_2$ denotes the Euclidean norm. For matrices, $\|\cdot\|_2$ and $\|\cdot\|_F$ denote the spectral and Frobenius norms. $\operatorname{span}(\cdot)$ denotes column space, $\mathbf I_N$ is the $N\times N$ identity matrix, and $\operatorname{tr}(\cdot)$, $\operatorname{diag}(\cdot)$, and $\operatorname{rank}(\cdot)$ denote trace, diagonal vector, and rank. $\mathbb E[\cdot]$ denotes expectation, and $\mathcal{CN}(\boldsymbol{\mu},\boldsymbol{\Sigma})$ denotes the circularly symmetric complex Gaussian distribution.

\section{System Model and Problem Formulation} \label{sec:sys_model}
We consider a single-cell downlink communication system, where a BS equipped with $N_t$ transmit antennas serves UEs deployed in a site indexed by $s$. The channel distribution is determinated by the radio propagation environments and therefore varies significantly across different sites. For each site, the channel is assumed to be block-fading, i.e., approximately constant within one coherence interval, during which SSB probing, CSI-RS transmission, and UE feedback are performed, followed by data transmission.
\vspace{-0.3cm}
\subsection{Channel and Signal Model}
The BS is assumed to employ a uniform linear array (ULA) with inter-element spacing $d$. We adopt the standard narrowband geometric channel model, where the channel is represented as a superposition of several dominant propagation paths \cite{Ayach2014SpatiallySparse}. For UE $q$ in site $s$, the channel vector can be expressed as follows
\begin{equation}\label{eq:channel_model_new}
    \mathbf h_{s,q}=\sum_{\ell=1}^{L_{s,q}} \alpha_{s,q,\ell}\mathbf a(u_{s,q,\ell}) = \mathbf{A}_{s,q}\boldsymbol{\alpha}_{s,q} \in\mathbb C^{N_t \times 1},
\end{equation}
where $L_{s,q}$ is the number of propagation paths, $\alpha_{s,q,\ell} \in\mathbb C$ is the complex path gain of the $\ell$-th path, and $\mathbf a(u_{s,q,\ell})$ is the corresponding normalized steering vector. For a ULA, the steering vector is given by
\begin{equation}
    \mathbf a(u)=\frac{1}{\sqrt{N_t}}\left[1,e^{j2\pi u},\ldots,e^{j2\pi(N_t-1)u}\right]^T \in\mathbb C^{N_t \times 1},
\end{equation}
where $u$ is the spatial frequency. Specifically, for the $\ell$-th path of UE $q$ in site $s$, we define $u_{s,q,\ell} \triangleq d \sin(\varphi_{s,q,\ell}) / \lambda$, where $\lambda$ is the carrier wavelength, and $\varphi_{s,q,\ell}$ is the angle of departure (AoD) from the BS to UE $q$. In addition, $\mathbf A_{s,q}=[\mathbf a(u_{s,q,1}),\ldots,\mathbf a(u_{s,q,L_{s,q}})] \in \mathbb{C}^{N_t \times L_{s,q}}$ and $\boldsymbol{\alpha}_{s,q}=[\alpha_{s,q,1},\ldots,\alpha_{s,q,L_{s,q}}]^T \in \mathbb{C}^{L_{s,q} \times 1}$ collect the steering vectors and path gains, respectively.

\begin{assumption}[\emph{Near-orthogonal angular paths}]
	\label{assump:near_orthogonal_paths}
	\normalfont
	The UE lies in the far field of the BS array, and the dominant path spatial frequencies $\{u_{s,q,\ell}\}_{\ell=1}^{L_{s,q}}$ are sufficiently separated relative to the array aperture. Therefore, the steering vectors of the dominant paths are approximately orthogonal, i.e.,
	\begin{equation}
		\mathbf A_{s,q}^H\mathbf A_{s,q} \approx \mathbf I_{L_{s,q}}.
		\label{eq:near_orthogonal_paths}
	\end{equation}
	This approximation is consistent with the favorable-propagation behavior of large antenna arrays with resolvable angular paths \cite{Ngo2014Favorable,Larsson2014MassiveMIMO}.
\end{assumption}

\indent Let $\mathbf{w}\in\mathbb{C}^{N_t \times 1}$ and $x$ denote the unit-power beamformer and symbol transmitted by the BS, respectively. Then, the received signal at the UE q in site s in given by
\begin{equation} \label{pdsch}
	y = \sqrt{P_s} \mathbf{h}_{s,q}^H\mathbf{w}x + n,
\end{equation}
where $P_s$ is the transmit power at site $s$, and $n\sim\mathcal{CN}(0,\sigma_{n,s}^2)$ is additive noise independent of the transmitted signal.
\vspace{-0.3cm}
\subsection{Beam Alignment Procedure}
In practical NR systems, beam alignment refers broadly to the procedure of identifying suitable beam directions or a low-dimensional beam-space representation for subsequent CSI acquisition and downlink transmission \cite{3GPP38214,3GPP38211,3GPP38215,Giordani2019BeamManagementNR}. The procedure typically starts with SSB beam sweeping, where the BS transmits a predefined set of synchronization beams $\mathcal{B}_s = \{\mathbf b_{s,k}\}_{k=1}^{K}$, and the UE measures their RSRP fingerprint $\mathbf r_{s,q}(\mathbf B_s)$ for initial access and coarse directional acquisition\footnote{Current NR systems typically rely on predefined SSB beams, e.g., DFT codebook. However, SSBF studies have shown that this codebook can be optimized w.r.t. the particular propagation environment in a certain site, hence we attach the site index to $\mathbf B_s$ to allow site-dependent SSB codebooks \cite{sim,Zhao2026UnifiedLF}.}. $\mathbf B_s=[\mathbf b_{s,1}, \ldots, \mathbf b_{s,K}]$ is a matrix collecting the SSB beams. The number of SSB beams $K$ is assumed fixed across sites, while their beam coefficients can vary across different site to adapt to the deployment environment \cite{sim}.
\\\indent
The SSB transmission follows the model in~\eqref{pdsch}, where $\mathbf{w}$ is specified as the SSB beams in $\mathbf B_s$ and $x$ denotes the transmitted SSB symbol. For each SSB beam, the UE averages the received signal power over the corresponding time-frequency SSB resources. The resulting RSRP value is therefore a beam-domain power measurement. Collecting the $K$ measurements over all SSB beams gives the RSRP fingerprint $\mathbf r_{s,q}(\mathbf B_s)$, which is reported by UE $q$ in a decibel (dB) form. Following \cite{sim}, the dB-domain RSRP fingerprint can be modeled as
\begin{equation} \label{RSRP}
	\mathbf{r}_{s, q} (\mathbf{B}_s) = \mathcal{M}(\mathbf{h}_{s,q},\mathbf{B}_s) + \mathbf{n}_{s,q} = \bar{\mathbf{r}}_{s, q} (\mathbf{B}_s) + \mathbf{n}_{s,q},
\end{equation}
where $\mathcal M(\cdot)$ denotes the RSRP measurement operator, $\bar{\mathbf{r}}_{s, q}(\mathbf{B}_s)$ is the noise-free RSRP, and $\mathbf{n}_{s,q}$ collects Gaussian noise. The detailed expressions of $\bar{\mathbf{r}}_{s, q}(\mathbf{B}_s)$ and $\mathbf{n}_{s,q}$ are given in Section II of \cite{sim}.
\\\indent
The RSRP fingerprint obtained at the SSB stage provides only coarse beam-level observations and is generally insufficient for high-resolution downlink precoding. To support data transmission, the system usually proceeds to a CSI refinement stage with three steps \cite{3GPP38214,3GPP38211,3GPP38215}: 1) the BS transmits CSI-RS on configured CSI-RS resources, 2) the UE estimates the downlink channel representation from the CSI-RS measurements, and 3) instead of feeding back the full channel vector, the UE reports a finite-dimensional CSI description, e.g., codebook indices, selected subspace, or low-dimensional coefficients, to support the BS's final beamforming decision due to limited channel resources. According to the analysis in \cite{Zhao2026UnifiedLF}, the effect of such a limited-feedback rule can be represented by the spatial subspace used for CSI-RS precoding. For a feedback rule $\psi$, let $\mathbf U_{s,q}^{(\psi)}$ denote the selected subspace basis vectors produced for UE $q$ in site $s$. The corresponding representing subspace is
\begin{equation}
	\mathcal U_{s,q}^{(\psi)} \triangleq \operatorname{span}(\mathbf U_{s,q}^{(\psi)}).
\end{equation}
Assuming that $\mathbf U_{s,q}^{(\psi)}$ has full column rank, the orthogonal projector onto this subspace is given by
\begin{equation}
	\mathbf P_{s,q}^{(\psi)} = \mathbf U_{s,q}^{(\psi)} \left( \left(\mathbf U_{s,q}^{(\psi)}\right)^{H}\mathbf U_{s,q}^{(\psi)} \right)^{-1} \left(\mathbf U_{s,q}^{(\psi)}\right)^{H}.
\end{equation}
The \emph{CSI-capture efficiency} of this subspace is defined as
\begin{equation}
	\eta_{s,q}^{(\psi)}(\mathbf h_{s,q}) \triangleq \frac{ \|\mathbf P_{s,q}^{(\psi)}\mathbf h_{s,q}\|_2^2}{\|\mathbf h_{s,q}\|_2^2}.
\end{equation}
This ratio is the normalized projection power of $\mathbf h_{s,q}$ onto the representing subspace selected by rule $\psi$. A value close to one means that the selected subspace preserves most of the channel direction, whereas a small value indicates severe channel-power loss before data precoding.
When maximum-ratio transmission (MRT) is applied within the selected subspace, the resulting beamforming vector is
\begin{equation}
	\hat{\mathbf w}_{s,q}^{(\psi)} = \frac{\mathbf P_{s,q}^{(\psi)}\mathbf h_{s,q}}{\|\mathbf P_{s,q}^{(\psi)}\mathbf h_{s,q}\|_2}.
\end{equation}
The corresponding signal-to-noise ratio (SNR) is $\rho_{s,q}\eta_{s,q}^{(\psi)}(\mathbf h_{s,q})$, where $\rho_{s,q}	\triangleq \frac{P_s\|\mathbf h_{s,q}\|_2^2}{\sigma_{n,s}^2}$ denotes the full-CSI SNR. Thus, for a fixed full-CSI SNR, the CSI-capture efficiency directly characterizes the SNR fraction preserved by feedback rule $\psi$.
\vspace{-0.3cm}
\subsection{Problem Formulation}
We design the limited feedback system to maximize the achievable data rate while keeping the associated overhead affordable. Accordingly, for UE $q$ in site $s$, we formulate the feedback design problem for rule $\psi$ as
\begin{equation}
	\label{prob:generic1}
	\max_{\psi\in\Psi} \  R_{s,q} (\psi) = \left(1-\frac{T_o(\psi)}{T}\right)	\log_2\!\Big(1+\rho_{s,q}\eta_{s,q}^{(\psi)}(\mathbf h_{s,q})\Big),
\end{equation}
where $\Psi$ is the feasible set of feedback rules, $T_o(\psi)$ denotes the channel-use overhead induced by rule $\psi$, and $T$ is the coherence interval length in channel uses. $T_o(\psi)$ usually consists of SSB/CSI-RS transmission/feedback overheads.
\\\indent
The optimization in \eqref{prob:generic1} couples CSI acquisition quality and signaling overhead through the same feedback rule $\psi$. In particular, richer measurements may improve the CSI-capture factor $\eta_{s,q}^{(\psi)}(\mathbf h_{s,q})$, but they also increase $T_o(\psi)$ and reduce the effective data-transmission time. Since the useful CSI subspace is channel dependent whereas full CSI is unavailable at the BS before feedback, improving \eqref{prob:generic1} requires a prior that links limited measurements to effective subspace decisions. Site-specific learning provides such a prior from target-site samples and is effective when one deployment confines channels to a limited propagation geometry. Scaling this idea to many sites, however, requires repeated labeled-data collection and training over diverse channel distributions, which becomes prohibitive at network scale. Wireless foundation models can address this challenge by learning transferable wireless representations from diverse deployments and adapting them to new sites with limited calibration data.

\section{SiFo for Site-Specific CSI Feedback} \label{sec:from_feedback_to_fm_sst2}
In this section, we present the proposed SiFo framework for solving \eqref{prob:generic1}. The framework combines the cross-site transfer capability of wireless foundation models with the local propagation adaptation of site-specific learning, thereby providing a data-efficient and low-overhead solution to \eqref{prob:generic1}. We first give the preliminary analysis that motivates SiFo, and then describe the proposed framework.
\vspace{-0.3cm}
\subsection{Preliminary Analysis}
\subsubsection{A Joint Analysis of Site-Agnostic and Site-Specific Feedback}
The feedback rule in \eqref{prob:generic1} can use different levels of site information. Conventional NR feedback represents the site-agnostic case, where one standardized codebook and feedback structure are applied to all sites. This rule is fixed before target-site propagation statistics are observed and corresponds to
\begin{equation}
	\psi_{\rm agn}^{\star} = \argmax_{\psi\in\Psi} \mathbb E_{s}\mathbb E_{q|s} \left[ R_{s,q}(\psi) \right],
	\label{eq:site_agnostic_solution}
\end{equation}
which maximizes the expected effective spectral efficiency over the population of sites and UEs. Type-I, Type-II, and PSC are representative standardized instantiations of this principle. Detailed analysis of these three schemes is provided in \cite{Zhao2026UnifiedLF}. Among them, Type-II feedback offers the richest subspace representation and highest spectral efficiency, but requires substantial online CSI acquisition, feedback overhead, and UE-side processing.
\\ \indent
Site-specific feedback allows the acquisition rule to depend on the deployment site. For a fixed site $s$, the ideal site-conditioned rule solves
\begin{equation}
	\psi_s^{\star} = \argmax_{\psi\in\Psi} \mathbb E_{q|s} \left[ R_{s,q}(\psi) \right].
	\label{eq:site_specific_solution}
\end{equation}
The dependence of $\psi_s^{\star}$ on $s$ represents the use of SSI, where the BS exploits target-site propagation statistics rather than applying one common rule to all sites.
\\ \indent
It is worth noting that site-agnostic design can be regarded as a special case of site-specific design by setting $\psi_s\equiv\psi$ for all $s$, leading to the following inequality
\begin{equation}
	\mathbb E_s \left[ \max_{\psi\in\Psi} \mathbb E_{q|s} \left[ R_{s,q}(\psi) \right] \right] \geq \max_{\psi\in\Psi} \mathbb E_s \mathbb E_{q|s} \left[ R_{s,q}(\psi) \right].
	\label{eq:site_specific_advantage}
\end{equation}
The inequality quantifies the target pursued by site-specific feedback: local propagation statistics can improve the rate objective by both increasing CSI-capture efficiency and reducing CSI acquisition overhead. Realizing this target requires an accurate site-conditioned rule for each deployment. The key difficulty is how to obtain such rules without collecting a large labeled dataset and retraining model for every new site.

\subsubsection{Single-Site SST2 as Site-Specific Feedback}
SST2 can be viewed as a learnable realization of the site-specific design in \eqref{eq:site_specific_solution}. Instead of requiring the UE to construct a high-resolution Type-II CSI feedback representation, SST2 shifts the subspace inference burden to the BS. Let $\boldsymbol{\Theta}=(\boldsymbol{\Theta}_{\rm B},\boldsymbol{\Theta}_{\rm H})$ collect the trainable SST2 parameters, where $\boldsymbol{\Theta}_{\rm B}$ parameterizes the probing codebook and $\boldsymbol{\Theta}_{\rm H}$ parameterizes the subspace inference module. An SST2 model can then be written as
\begin{equation}
	\psi_{\boldsymbol{\Theta}}^{\rm SST2} = \left( \mathbf B_{\boldsymbol{\Theta}_{\rm B}},\, H_{\boldsymbol{\Theta}_{\rm H}} \right),
	\label{eq:sst2_model}
\end{equation}
where $\mathbf B_{\boldsymbol{\Theta}_{\rm B}}$ generates a low-dimensional RSRP fingerprint and $H_{\boldsymbol{\Theta}_{\rm H}}$ maps this fingerprint to a representing subspace. When the distinction between the two parameter blocks is not essential, we use the shorthand $\mathbf B_{\boldsymbol{\Theta}}$ and $H_{\boldsymbol{\Theta}}$. For UE $q$ in site $s$, following model~\eqref{RSRP}, the measured fingerprint is given by
\begin{equation}
	\mathbf r_{s,q} (\mathbf B_{\boldsymbol{\Theta}}) = \mathcal M(\mathbf h_{s,q},\mathbf B_{\boldsymbol{\Theta}}) + \mathbf{n}_{s,q},
	\label{eq:sst2_fingerprint}
\end{equation}
and the BS predicts the representing subspace as
\begin{equation}
	\hat{\mathcal U}_{s,q} = H_{\boldsymbol{\Theta}}(\mathbf r_{s,q} (\mathbf B_{\boldsymbol{\Theta}})).
	\label{eq:sst2_subspace_prediction}
\end{equation}
Given a prescribed probing budget, the overhead term in \eqref{prob:generic1} is constant across SST2 parameter choices, and the achievable-rate objective becomes monotone in the CSI-capture efficiency. Therefore, the single-site SST2 baseline can be trained through the following capture-efficiency objective:
\begin{equation}
	\boldsymbol{\Theta}_s^\star = \argmax_{\boldsymbol{\Theta}} \mathbb E_{\mathbf h\sim\mathcal P_s} \left[ \eta \left( \mathbf h,\psi_{\boldsymbol{\Theta}}^{\rm SST2} \right) \right],
	\label{eq:single_site_sst2_objective}
\end{equation}
where $\mathcal P_s$ denotes the channel distribution of site $s$.
\\ \indent
This formulation explains why SST2 can be effective. The model only needs to learn the conditional relation induced by one fixed propagation distribution $\mathcal P_s$ rather than to solve a universal RSRP-to-subspace mapping problem. With sufficient data from the same site, SSI supervision resolves the ambiguity between coarse amplitude-only RSRP fingerprints and high-dimensional subspace decisions.
\\ \indent
This formulation also exposes the limitation of SST2. The learned parameters in \eqref{eq:single_site_sst2_objective} are tied to one deployment distribution. For a new target site $t$ with distribution $\mathcal P_t$, the optimal SST2 parameters may be different:
\begin{equation}
	\boldsymbol{\Theta}_t^\star = \argmax_{\boldsymbol{\Theta}} \mathbb E_{\mathbf h\sim\mathcal P_t} \left[ \eta \left( \mathbf h,\psi_{\boldsymbol{\Theta}}^{\rm SST2} \right) \right].
	\label{eq:target_site_sst2_objective}
\end{equation}
Training such a site-specific model from scratch for every target site requires abundant target-site data and repeated training, which limits the scalability of SST2.

\subsubsection{Cross-Site Pretraining Limitation}
A natural way to improve scalability is to pretrain one SST2 model over multiple source sites. The corresponding training objective maximizes the average CSI-capture efficiency over the source-site channel distributions:
\begin{equation}
	\boldsymbol{\Theta}_0^\star = \argmax_{\boldsymbol{\Theta}} \mathbb E_{s\in\mathcal S_{\rm src}} \mathbb E_{\mathbf h\sim\mathcal P_s} \left[ \eta \left( \mathbf h,\psi_{\boldsymbol{\Theta}}^{\rm SST2} \right) \right].
	\label{eq:cross_site_pretraining_objective}
\end{equation}
This objective produces a single shared SST2 model trained over the source-site mixture. If deployed directly to a new target site, however, the model applies the same rule $\psi_{\boldsymbol{\Theta}_0^\star}^{\rm SST2}$ to all target deployments. Its decision is conditioned on the observed RSRP fingerprint, but not on the local propagation statistics of the target site.
\\ \indent
The limitation comes from the ambiguity in inferring CSI subspaces from RSRP fingerprints. For this ambiguity argument, let $\mathbf P_Q^\star$ denote the full-CSI rank-$Q$ reference projector associated with the prescribed CSI representing subspace, where $Q$ is the representation dimension fixed by the system requirement. For a fixed fingerprint $\mathbf r$, different sites may induce different conditional rank-$Q$ subspace statistics, i.e.,
\begin{equation}
	\mathbb E \left[ \mathbf P_Q^\star \,\middle|\,	\mathbf r, s \right] \neq \mathbb E \left[ \mathbf P_Q^\star \,\middle|\, \mathbf r, t \right],
	\label{eq:rsrp_subspace_ambiguity}
\end{equation}
because similar amplitude-only RSRP fingerprints can arise from different multipath geometries, angular spreads, phase relations, and scattering structures. Thus, cross-site pretraining can provide a useful marginal prior, but it cannot by itself determine the target-site conditional relation in \eqref{eq:rsrp_subspace_ambiguity}. Limited fine-tuning may partially specialize the parameters, but it still compresses the target-site samples into model weights and does not retain sample-level local evidence during inference.
\vspace{-0.3cm}
\subsection{Proposed SiFo Framework}\label{subsec:fm_sst2_framework_overview}
SiFo restores target-site conditioning at inference time. Cross-site pretraining learns a shared SSB probing codebook and a marginal fingerprint-to-subspace prior over source deployments. Target-site calibration then stores local full-CSI subspace evidence as projector-valued memory, so that online UEs can be interpreted relative to calibrated UEs from the same deployment. SiFo operates through three deployment stages as shown in Fig. \ref{fig:arch}:
\begin{itemize}
	\item \textbf{Stage 1: Cross-site pretraining.} The BS trains an SST2 model on multiple source sites $\mathcal S_{\rm src}$, yielding pretrained parameters $\boldsymbol{\Theta}_0^\star$. The resulting model contains a learned probing codebook $\mathbf B_{\boldsymbol{\Theta}_0^\star}$ and a subspace inference network $H_{\boldsymbol{\Theta}_0^\star}(\cdot)$ that maps an RSRP fingerprint to a $Q$-dimensional CSI representation basis. This stage provides shared RSRP fingerprint coordinates and a transferable parametric prior for CSI subspace acquisition.
	\item \textbf{Stage 2: Target-site projector calibration.} For a new target site $t$, the BS collects a calibration UE set $\mathcal C_t$. Each calibration UE is measured through the pretrained SSB probing codebook and is associated offline with a full-CSI direction projector. These calibrated samples form the target-site projector memory used by the online acquisition rule.
	\item \textbf{Stage 3: Calibration-aided online subspace acquisition.} For served UE $q$, the UE reports the same $K$ RSRP measurements produced by the pretrained SSB probing codebook. The BS retrieves calibrated UEs with similar fingerprints from the target-site memory, forms a memory-based projector estimate from their full-CSI direction projectors, and fuses it with the parametric subspace prediction through a UE-dependent confidence weight. The fused estimate is then converted into the final CSI representation subspace.
\end{itemize}

\begin{figure}[t]
	\centering
	\includegraphics[scale=0.37]{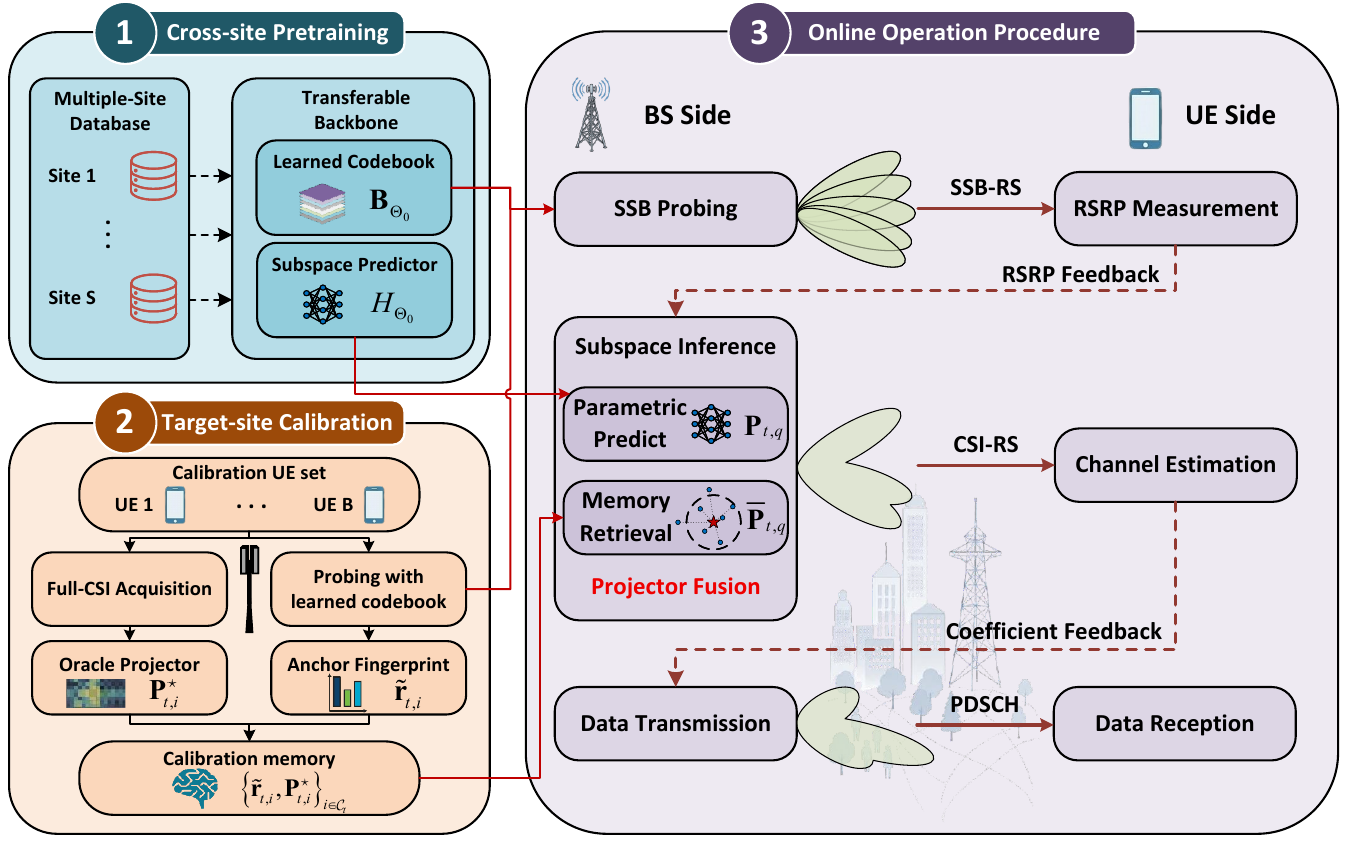}
	\caption{Illustration of the SiFo framework}
	\label{fig:arch}
\end{figure}

When deployed to a new site, SiFo collects a small size of target-site labels to construct projector memory while the parametric branch remains fixed, i.e., $\boldsymbol{\Theta}_t=\boldsymbol{\Theta}_0^\star$. The calibration principle is developed in Section~\ref{sec:sensing_criteria}. The projector-memory construction and online fusion rule are detailed in Section~\ref{sec:site_sensing_encoding}. The backbone architecture and deployment implementation are specified in Section~\ref{sec:backbone_deployment}.

\section{Low-Overhead Target-Site Calibration}\label{sec:sensing_criteria}
In this section, we develop the low-overhead target-site calibration principle used in Stage 2 of SiFo. The goal is to identify what compact RSRP measurements can provide for Type-II subspace acquisition and why full-CSI-based projector labels are needed during calibration.
\vspace{-0.3cm}
\subsection{Calibration RSRP and Directional-Power Statistics}\label{subsec:what_site_sensing_estimates}
Target-site calibration uses RSRP measurements to sense spatial power over a deployment. The analysis involves two spatial covariance levels: a UE-level covariance for the local CSI representation subspace and a site-level covariance for aggregate directional-power statistics. Throughout this section, channel covariances refer to uncentered spatial second-order moments. Let $q\sim\mathcal Q_t$ denote a UE in target site $t$, and let $\mathbf h\sim\mathcal P_t(\cdot|q)$ denote its local downlink channel distribution. The UE-level channel covariance is defined as
\begin{equation}
	\mathbf R_{t,q}^{\rm UE} \triangleq \mathbb E_{\mathbf h\sim\mathcal P_t(\cdot|q)} \left[ \mathbf h\mathbf h^H \right].
	\label{eq:ue_level_covariance}
\end{equation}
The site-level channel covariance is the mixture of UE-level covariances over the deployment, which is given by
\begin{equation}
	\mathbf R_t = \mathbb E_{q\sim\mathcal Q_t} \left[ \mathbf R_{t,q}^{\rm UE} \right] = \mathbb E_{\mathbf h\sim\mathcal P_t} [\mathbf h\mathbf h^H].
	\label{eq:site_ue_covariance_mixture}
\end{equation}
Here, $\mathcal P_t$ denotes the marginal target-site channel distribution. Thus, $\mathbf R_t$ describes the aggregate, power-weighted spatial structure sensed over the site, while $\mathbf R_{t,q}^{\rm UE}$ describes the local spatial statistics relevant to UE $q$.
\\ \indent
Under the Type-II feedback budget, the BS ultimately selects a prescribed $Q$-dimensional CSI representation subspace for each UE. If the UE-level covariance were available, the full-CSI reference subspace for UE $q$ would be the $Q$-dimensional eigenspace that captures the largest average channel power. Large-scale gain affects the received SNR but not this spatial eigenspace. Therefore, the following normalized UE-level covariance is adopted:
\begin{equation}
	\widetilde{\mathbf R}_{t,q}^{\rm UE} \triangleq \frac{\mathbf R_{t,q}^{\rm UE}} {\operatorname{tr}(\mathbf R_{t,q}^{\rm UE})},
	\label{eq:normalized_ue_second_moment}
\end{equation}
which preserves the eigenspaces of $\mathbf R_{t,q}^{\rm UE}$ while removing power scaling. The corresponding full-CSI reference subspace is the dominant $Q$-dimensional eigenspace of $\widetilde{\mathbf R}_{t,q}^{\rm UE}$. The site-level covariance $\mathbf R_t$ governs directional-power statistics over the deployment, whereas $\widetilde{\mathbf R}_{t,q}^{\rm UE}$ determines the local CSI representation subspace for UE $q$.
\\ \indent
Let $\mathbf S=[\mathbf s_1,\ldots,\mathbf s_J]\in\mathbb C^{N_t\times J}$ denote a generic RSRP sensing codebook. For calibration UE $i\in\mathcal C_t$, the noise-free linear-domain RSRP fingerprint induced by $\mathbf S$ is
\begin{equation}
	\bar{\mathbf r}_{t,i}(\mathbf S) = \left[ |\mathbf s_1^H\mathbf h_{t,i}|^2,\ldots, |\mathbf s_J^H\mathbf h_{t,i}|^2 \right]^T = |\mathbf S^H\mathbf h_{t,i}|^2,
	\label{eq:site_power_profile_general}
\end{equation}
where the absolute value and square operations are performed element-wise. The corresponding population directional-power statistic is
\begin{equation}
	\boldsymbol{\pi}_t(\mathbf S) = \mathbb E_{\mathbf h\sim\mathcal P_t} [ |\mathbf S^H\mathbf h|^2] = \operatorname{diag} \left( \mathbf S^H\mathbf R_t\mathbf S \right),
	\label{eq:site_stat_general}
\end{equation}
whose $\ell$-th entry can be expressed as $[\boldsymbol{\pi}_t(\mathbf S)]_\ell = \mathbf s_\ell^H\mathbf R_t\mathbf s_\ell$. The $\ell$-th entry therefore represents the site-average received power along sensing direction $\mathbf s_\ell$. Thus, calibration RSRP observes diagonal directional-power projections of $\mathbf R_t$ in the sensing coordinates. This is the spatial information supplied by RSRP: sample-level beam-power fingerprints and their site-level directional-power statistics.
\\ \indent
In SiFo, the generic sensing codebook $\mathbf S$ is instantiated by the pretrained learned SSB codebook $\mathbf B_{\boldsymbol{\Theta}_0^\star}$ for both calibration UEs and online served UEs. Using the same SSB probing codebook places calibration fingerprints and online query fingerprints in the same beam-power coordinate system, so nearest-neighbor retrieval compares responses generated by the same probing beams.
\vspace{-0.3cm}
\subsection{Prior-Free Sensing Principle}\label{subsec:prior_free_isotropic_sensing}
RSRP sensing observes directional projections of the unknown site-level channel covariance $\mathbf R_t$. If $\mathbf R_t$ were available, the sensing beams could be designed to emphasize the dominant site-specific angular subspace. Before target-site calibration, however, neither $\mathcal P_t$ nor $\mathbf R_t$ is known to the BS. This motivates a prior-free reference criterion for sensing design: before observing SSI, the sensing codebook should not privilege any site-specific angular sector.

\begin{assumption}[\emph{Isotropic sensing prior for unknown sites}] \normalfont
	\label{assump:prior_free_site_sensing}
	Before SSI is collected, no unit-norm direction in $\mathbb C^{N_t}$ is regarded as more likely or more important than any other. The calibration-stage sensing codebook should therefore provide directionally uniform observability and avoid blind spatial directions, rather than optimize for any presumed site-dependent angular sector.
\end{assumption}

The isotropic prior is epistemic rather than physical: real deployment sites are generally anisotropic once geometry and scatterers are fixed. It only states that, before SSI is observed, the sensing design itself should not encode a site-specific angular preference.
For a unit-norm channel direction $\mathbf h$, the sensing energy collected by $\mathbf S$ is $\|\mathbf S^H\mathbf h\|_2^2=\mathbf h^H\mathbf S\mathbf S^H\mathbf h$. Therefore, the prior-free reference design that maximizes the worst-case sensing energy can be formulated into the following optimization problem:
\begin{subequations}
	\label{prob:isotropic_site_sensing}
	\begin{align}
		\max_{\mathbf S}\quad &\min_{\|\mathbf h\|_2=1} \|\mathbf S^H\mathbf h\|_2^2 \\
		\mathrm{s.t.}\quad &\|\mathbf s_j\|_2=1,\qquad j=1,\ldots,J,\\
		&\mathbf s_j\in\mathcal A_{\rm RF},\qquad j=1,\ldots,J,
	\end{align}
\end{subequations}
where the unit-norm constraint fixes the per-beam sensing power, and $\mathcal A_{\rm RF}$ captures analog-beam feasibility constraints such as phase-only or quantized-phase implementations.
\\ \indent 
The sensing-budget consequence follows from the rank of the sensing covariance. With fewer than $N_t$ sensing beams, the rank satisfies
\begin{equation}
	\operatorname{rank}(\mathbf S\mathbf S^H)\leq J<N_t.
\end{equation}
As a consequence, the sensing beams span at most a $J$-dimensional subspace of $\mathbb C^{N_t}$. Hence, there exists a unit-norm direction $\mathbf h_0\in\operatorname{span}\{\mathbf s_1,\ldots,\mathbf s_J\}^{\perp}$ for which $\mathbf S^H\mathbf h_0=\mathbf 0$ and hence
\begin{equation}
	\min_{\|\mathbf h\|_2=1}\|\mathbf S^H\mathbf h\|_2^2=0 .
\end{equation}
Thus, no prior-free sensing design with fewer than $N_t$ beams can avoid blind spatial directions in the worst case. The case $J=N_t$ is therefore the minimum non-degenerate full-coverage reference, and it is characterized below.

\begin{theorem}[\emph{Minimum-budget full-coverage sensing reference}] \normalfont
	\label{thm:min_budget_isotropic_sensing}
	Consider \eqref{prob:isotropic_site_sensing} with $J=N_t$ unit-norm sensing beams. The optimal worst-case sensing energy is one. Moreover, a feasible codebook attains this value if and only if $\mathbf S\mathbf S^H=\mathbf I_{N_t}$. Consequently, any feasible unitary sensing codebook is optimal. For a ULA with $N_t$ antennas, the DFT-$N_t$ codebook $\mathbf S^\star=\mathbf D_{N_t}$ is one such optimal solution, where
	\begin{equation}
		[\mathbf D_{N_t}]_{n,k} = \frac{1}{\sqrt{N_t}} e^{-j\frac{2\pi}{N_t}nk}, \quad n,k=0,\ldots,N_t-1.
		\label{eq:dft_codebook_def}
	\end{equation}
\end{theorem}
\begin{IEEEproof}
	See Appendix~\ref{app1}.
\end{IEEEproof}

\begin{remark}[\emph{Implication for learned low-overhead probing}] \normalfont
	\label{rem:learned_codebook_sensing_principle}
	\textbf{Theorem~\ref{thm:min_budget_isotropic_sensing}} gives a prior-free full-coverage reference: at least $N_t$ beams are needed to avoid blind spatial directions, and the non-redundant solution is orthogonal sensing. The DFT-$N_t$ codebook is one convenient ULA realization, but deploying $N_t$ SSB beams would impose excessive synchronization-stage overhead. SiFo therefore uses a low-dimensional learned codebook with $K\ll N_t$, accepting that worst-case full-space coverage is no longer possible. In this finite-budget regime, the probing beams should remain orthogonal within the selected $K$-beam sensing subspace and produce distinguishable RSRP angular responses for downstream subspace acquisition. The former is measured by the off-diagonal energy of the Gram matrix of $\mathbf B_{\boldsymbol{\Theta}_0^\star}$, i.e., $\mathbf B_{\boldsymbol{\Theta}_0^\star}^H\mathbf B_{\boldsymbol{\Theta}_0^\star}$. Section~\ref{subsec:key_ablation} shows that the learned codebook is naturally near-orthogonal and approaches the performance of a wider DFT reference with significantly reduced overhead.
\end{remark}

\begin{remark}[\emph{Angular interpretation of directional-power sensing}]\normalfont
	\label{rem:angular_interpretation_power_sensing}
	For a ULA, directional-power sensing admits an angular-kernel interpretation. The DFT reference corresponds to fixed Dirichlet kernels on a uniform angular grid, yielding an angular-bin power profile with leakage and multi-cluster peaks at the UE-sample level. At the site level, the corresponding population statistic is $\operatorname{diag}(\mathbf D_{N_t}^H\mathbf R_t\mathbf D_{N_t})$, which represents a DFT-smoothed angular power spectrum of the target-site covariance. A learned probing codebook replaces the fixed DFT grid with data-adapted directional kernels, so $\operatorname{diag}(\mathbf B_{\boldsymbol{\Theta}_0^\star}^H\mathbf R_t\mathbf B_{\boldsymbol{\Theta}_0^\star})$ gives the same site-level directional-power statistic measured in the learned sensing coordinates. The SiFo fingerprint is therefore a low-dimensional sample-level realization of these learned directional-power projections.
\end{remark}
\vspace{-0.5cm}
\subsection{Full-CSI Direction Labels for Calibration}\label{subsec:why_projector_labels}
The prior-free reference clarifies the limitation of low-dimensional probing. In SiFo, the deployed codebook $\mathbf{B}_{\boldsymbol{\Theta}_0^\star}$ uses $K \ll N_t$ learned beams, so the resulting RSRP vector is a compact directional-power fingerprint rather than an invertible channel observation. Some spatial directions are inevitably blind or weakly observed under a worst-case full-space criterion.
\\ \indent
This compact fingerprint remains useful for indexing calibrated UEs. Target-site propagation is spatially structured, so UEs with similar learned-beam power responses tend to share related dominant angular support, angular spread, or local scattering conditions. With the shared learned codebook, calibration and served UEs are measured in the same beam-power coordinates, and nearby calibration samples provide local references for the served UE.
\\ \indent
The RSRP coordinate alone does not specify the associated CSI subspace. Even under the $N_t$-beam DFT reference, RSRP contains only beam-domain magnitudes for an instantaneous channel or diagonal directional-power statistics at the population level. The off-diagonal covariance information needed to determine CSI subspace geometry remains unobserved. SiFo therefore stores calibration samples in a paired form: the RSRP fingerprint provides the directional-power coordinate, while an offline full-CSI direction projector provides the corresponding subspace label. This paired calibration information enables target-site subspace acquisition with low-dimensional learned $K$-beam online probing.

\section{Calibration-Aided CSI Subspace Acquisition}
 \label{sec:site_sensing_encoding}
In this section, we develop the calibration-aided CSI subspace acquisition rule used in Stage 3 of SiFo. The target-site calibration output is a set of projector-labeled RSRP fingerprints, which supports online acquisition through three operations: projector-labeled sample construction, local covariance acquisition from neighboring calibrated UEs, and confidence-weighted fusion with the parametric predictor.
\vspace{-0.3cm}
\subsection{Projector-Labeled Calibration Memory}\label{subsec:calibration_memory_construction}
The calibration memory preserves the sample-level relation between a low-dimensional RSRP fingerprint and the corresponding full-CSI reference subspace. With the pretrained $K$-beam SSB codebook $\mathbf B_{\boldsymbol{\Theta}_0^\star}$, calibration UE $i$ obtains $\mathbf r_{t,i}$ following~\eqref{RSRP}. For calibration-neighbor search, the relevant information is the shape of the beam-power response across the learned SSB beams. This shape reflects the UE's angular support, angular spread, and local scattering structure, whereas the total RSRP level is dominated by large-scale gain and does not change the desired subspace. Since the projector label below is also gain-invariant, SiFo uses the normalized RSRP profile as the calibration key, which is given by
\begin{equation}
	\tilde{\mathbf r}_{t,i} = \frac{\mathbf r_{t,i}} {\|\mathbf r_{t,i}\|_2} \in \mathbb{R}^{K \times 1}.
	\label{eq:calibration_key}
\end{equation}
Because calibration and online UEs are measured by the same probing codebook, their normalized fingerprints are comparable in a common learned beam-power coordinate system.
\\ \indent
The RSRP key identifies local similarity under the shared probing codebook, but it is not itself a CSI subspace decision. The calibration memory is hence defined as a set of key-value pairs, where the key is the normalized RSRP profile and the value is the gain-normalized full-CSI direction projector:
\begin{equation}
	\mathcal{Y}_t = \left\{ (\tilde{\mathbf r}_{t,i},\mathbf P_{t,i}^{\star}) \right\}_{i\in\mathcal C_t}, \quad \mathbf P_{t,i}^{\star} = \frac{\mathbf h_{t,i}\mathbf h_{t,i}^H}{\|\mathbf h_{t,i}\|_2^2}.
	\label{eq:full_csi_direction_projector}
\end{equation}
The projector $\mathbf P_{t,i}^{\star}$ records the normalized spatial direction observed from the full-CSI vector of calibration UE $i$. Since one calibration channel vector provides only one spatial direction, this label is a direction-level covariance sample rather than the final rank-$Q$ CSI representation subspace. The prescribed $Q$-dimensional subspace is formed later by averaging the projectors of neighboring calibrated UEs in \eqref{ensemble} and extracting the dominant rank-$Q$ eigenspace in \eqref{eq:final_projector}. For any candidate precoding vector $\mathbf v$, the quadratic form $\mathbf v^H\mathbf P_{t,i}^{\star}\mathbf v$ measures how much of $\mathbf v$ lies along the true channel direction after channel gain and common phase are removed. Hence, $\mathbf P_{t,i}^{\star}$ supplies the phase-aware spatial-direction information that the amplitude-only RSRP key $\tilde{\mathbf r}_{t,i}$ cannot provide.
\begin{lemma}[\emph{Full-CSI direction sample and local covariance}]
	\label{lem:full_csi_direction_covariance}\normalfont
	Under the sparse geometric channel model and \textbf{Assumption~\ref{assump:near_orthogonal_paths}}, further assume that the instantaneous channel power concentrates around its UE-level average, i.e.,
	\begin{equation}
		\|\mathbf h_{t,i}\|_2^2 \approx	\mathbb E_{\mathbf h_{t,i}\sim\mathcal P_{t,i}^{\rm UE}} \!\left[\|\mathbf h_{t,i}\|_2^2\right] = \operatorname{tr}(\mathbf R_{t,i}^{\rm UE}),
	\end{equation}
	where $\mathcal P_{t,i}^{\rm UE}$ denotes the local channel distribution of calibration UE $i$. Then the full-CSI direction projector satisfies
	\begin{equation}
		\mathbb E_{\mathbf h_{t,i}\sim\mathcal P_{t,i}^{\rm UE}} \!\left[\mathbf P_{t,i}^{\star}\right] \approx \widetilde{\mathbf R}_{t,i}^{\rm UE},
		\label{eq:prop1}
	\end{equation}
	where $\widetilde{\mathbf R}_{t,i}^{\rm UE}$ is defined by \eqref{eq:normalized_ue_second_moment}.
\end{lemma}

\begin{IEEEproof}
	See Appendix~\ref{app:lem_full_csi_direction_covariance}.
\end{IEEEproof}
\textbf{Lemma~\ref{lem:full_csi_direction_covariance}} establishes the statistical link between the stored full-CSI direction labels and the local normalized covariance. Each calibration entry therefore provides a phase-aware spatial-direction observation, which becomes informative for Type-II subspace extraction after aggregation over neighboring entries with similar RSRP fingerprints.

\subsection{Local Channel-Covariance Acquisition}\label{subsec:memory_retrieval_fusion}
For a served UE $q$ in target site $t$, the BS obtains its normalized RSRP fingerprint $\tilde{\mathbf r}_{t,q}$ from SSB probing following~\eqref{RSRP}. The site-average covariance $\mathbf R_t$ is coarse for UE-level subspace acquisition because it averages UEs with different angular supports and scattering conditions. The relevant statistic is instead a covariance conditioned on the served UE's normalized RSRP fingerprint, which selects calibration UEs with similar learned-beam power responses. Since large-scale channel power does not change the desired CSI representation subspace, this local statistic is formed after channel-gain normalization. The resulting fingerprint-conditioned local covariance is defined as
\begin{equation}
	\mathbf C_t(\tilde{\mathbf r}_{t,q}) \triangleq \mathbb E_{\mathbf h\sim\mathcal P_t} \left[ \frac{\mathbf h\mathbf h^H}{\|\mathbf h\|_2^2} \,\middle|\, \tilde{\mathbf r}(\mathbf h)\simeq\tilde{\mathbf r}_{t,q} \right].
	\label{eq:local_projector_moment}
\end{equation}
The site-level average of these gain-normalized channel directions removes the channel-power weighting retained by $\mathbf R_t$ in \eqref{eq:site_ue_covariance_mixture}, which is given by
\begin{equation}
	\mathbb E_{\tilde{\mathbf r}} \left[ \mathbf C_t(\tilde{\mathbf r}) \right] = \mathbb E_{\mathbf h\sim\mathcal P_t} \left[ \frac{\mathbf h\mathbf h^H}{\|\mathbf h\|_2^2} \right],
	\label{eq:directional_moment_marginal}
\end{equation}
\begin{assumption}[\emph{Fingerprint-local channel stationarity}]
	\label{assump:stationarity}\normalfont
	Let $\mathcal N(q)\subseteq\mathcal C_t$ denote a sufficiently small neighborhood of the served UE $q$ in the normalized RSRP fingerprint space. For any calibration UE $q'\in\mathcal N(q)$, the normalized UE-level channel covariance varies slowly within this neighborhood, i.e.,
	\begin{equation}
		\widetilde{\mathbf R}_{t,q'}^{\rm UE} \approx \widetilde{\mathbf R}_{t,q}^{\rm UE}.
	\end{equation}
\end{assumption}

Under \textbf{Assumption~\ref{assump:stationarity}}, the local covariance in \eqref{eq:local_projector_moment} has the same dominant subspace as the UE-level normalized covariance, i.e., 
\begin{equation}
	\mathbf C_t(\tilde{\mathbf r}_{t,q}) \approx \widetilde{\mathbf R}_{t,q}^{\rm UE}.
	\label{eq:memory_local_second_moment}
\end{equation}
Let $m$ denote the neighborhood size and $\mathcal N_m(q) \subseteq \mathcal C_t$ denote the nearest calibration indices to $\tilde{\mathbf r}_{t,q}$ by cosine similarity. The local covariance estimate is acquired by averaging their projector labels, which is given by
\begin{equation} \label{ensemble}
	\hat{\mathbf P}_{t,q}^{(m)}(\tilde{\mathbf r}_{t,q}) = \frac{1}{m} \sum_{i\in\mathcal N_m(q)} \mathbf P_{t,i}^{\star}.
\end{equation}
By \textbf{Lemma~\ref{lem:full_csi_direction_covariance}}, each term in~\eqref{ensemble} is a gain-normalized covariance sample. By the local stationarity assumption, neighboring terms correspond to similar local spatial statistics. Therefore, $\hat{\mathbf P}_{t,q}^{(m)}(\tilde{\mathbf r}_{t,q})$ is an empirical estimate of $\mathbf C_t(\tilde{\mathbf r}_{t,q})$.

\begin{theorem}[\emph{Consistency of local covariance acquisition}]\label{thm:memory_covariance_consistency}\normalfont
Under \textbf{Assumption~\ref{assump:stationarity}}, consider a sequence of calibration sets with $|\mathcal C_t|\to\infty$ and neighborhood sizes satisfying $m(|\mathcal C_t|)\to\infty$ and $m/|\mathcal C_t|\to0$. The local projector average in \eqref{ensemble} is then asymptotically unbiased for the UE-level normalized covariance, i.e.,
\begin{equation}
	\mathbb E\!\left[ \hat{\mathbf P}_{t,q}^{(m)}(\tilde{\mathbf r}_{t,q}) \right] \xrightarrow{|\mathcal C_t|\to\infty} \widetilde{\mathbf R}_{t,q}^{\rm UE}.
\end{equation}
The same limit holds for any finite average of neighborhood sizes satisfying the same scaling.
\end{theorem}
\begin{IEEEproof}
	See Appendix~\ref{app:thm_memory_covariance_consistency}.
\end{IEEEproof}

\textbf{Theorem~\ref{thm:memory_covariance_consistency}} formalizes the role of calibration coverage: as calibrated samples become dense in the target-site fingerprint space, projector averaging recovers the local normalized channel covariance associated with the served UE. In finite calibration sets, the neighborhood size controls a variance--locality tradeoff. SiFo therefore averages several neighborhood estimates,
\begin{equation}
	\bar{\mathbf P}_{t,q}(\tilde{\mathbf r}_{t,q}) = \frac{1}{|\mathcal S_m|} \sum_{m\in\mathcal S_m} \hat{\mathbf P}_{t,q}^{(m)}(\tilde{\mathbf r}_{t,q}),
	\label{eq:multi_scale_memory}
\end{equation}
where $\mathcal S_m$ is a finite set of neighborhood sizes specified in Section~\ref{sec:results}.

\subsection{Confidence-Weighted Fusion and Subspace Extraction}\label{subsec:fusion_subspace_extraction}
With a finite calibration set, the memory estimate is reliable only when the served UE has close calibrated neighbors in fingerprint space. When no close neighbor exists, the pretrained SST2 subspace inference network provides a complementary parametric estimate $\tilde{\mathbf P}_{t,q}(\mathbf r_{t,q},\boldsymbol{\Theta}_t)$, which represents the cross-site learned subspace prior conditioned on the same RSRP observation. For SiFo deployment, $\boldsymbol{\Theta}_t=\boldsymbol{\Theta}_0^\star$. Fine-tuned parameters are used only in the FT-SST2 baseline. The final acquisition rule therefore fuses the parametric and memory estimates through a UE-dependent confidence weight $\alpha_q\in[0,1]$:
\begin{equation}\label{eq:pmix}
	\hat{\mathbf P}_{t,q}^{\rm mix} = (1-\alpha_q) \tilde{\mathbf P}_{t,q}(\mathbf r_{t,q},\boldsymbol{\Theta}_t) + \alpha_q \bar{\mathbf P}_{t,q}(\tilde{\mathbf r}_{t,q}).
\end{equation}
Although the memory labels are projector-valued, the multi-scale average $\bar{\mathbf P}_{t,q}$ and the fused matrix $\hat{\mathbf P}_{t,q}^{\rm mix}$ are covariance estimates before rank reduction. To state the resulting subspace constraint, let
\begin{equation}
	\mathcal P_Q \triangleq \left\{ \mathbf P: \mathbf P=\mathbf P^H,\, \mathbf P^2=\mathbf P,\, \operatorname{rank}(\mathbf P)=Q \right\}
	\label{eq:rank_q_projector_set}
\end{equation}
denote the set of rank-$Q$ orthogonal projectors. After projector averaging and fusion, $\bar{\mathbf P}_{t,q}\notin\mathcal P_Q$ and $\hat{\mathbf P}_{t,q}^{\rm mix}\notin\mathcal P_Q$ in general: they generally lose idempotence and may have rank larger than $Q$. Instead, they represent gain-normalized spatial covariance estimates, from which the protocol-constrained rank-$Q$ CSI subspace is extracted. For any Hermitian matrix $\mathbf A$, let $\boldsymbol\Pi_Q(\mathbf A)\in\mathcal P_Q$ denote any projector onto a dominant $Q$-dimensional eigenspace of $\mathbf A$. Under the prescribed $Q$-dimensional CSI representing-subspace constraint, the final CSI subspace decision is
\begin{equation}\label{eq:final_projector}
	\hat{\mathbf P}_{t,q}^{Q} \triangleq \boldsymbol\Pi_Q\!\left(\hat{\mathbf P}_{t,q}^{\rm mix}\right) = \mathbf U_Q\mathbf U_Q^H,
\end{equation}
where the columns of $\mathbf U_Q\in\mathbb C^{N_t\times Q}$ are the $Q$ eigenvectors of $\hat{\mathbf P}_{t,q}^{\rm mix}$ associated with its $Q$ largest eigenvalues.

\begin{lemma}[\emph{Near-reference rank-$Q$ subspace extraction}]\label{lem:near_reference_rank_q}\normalfont
Define the UE full-CSI Type-II reference projector as
\begin{equation}
	\mathbf P_{t,q}^{\rm ref,Q}
	\triangleq
	\boldsymbol\Pi_Q
	\left(
	\widetilde{\mathbf R}_{t,q}^{\rm UE}
	\right),
	\label{eq:full_csi_rank_q_projector}
\end{equation}
which is the rank-$Q$ subspace that would be selected if the UE-level normalized covariance were known. The projector $\hat{\mathbf P}_{t,q}^{Q}$ in \eqref{eq:final_projector} satisfies
\begin{equation}
	\hat{\mathbf P}_{t,q}^{Q}
	\in
	\argmax_{\mathbf P\in\mathcal P_Q}
	\operatorname{tr}\!\left(
	\mathbf P\hat{\mathbf P}_{t,q}^{\rm mix}
	\right),
	\label{eq:rank_q_power_optimal}
\end{equation}
and equivalently gives the closest element of $\mathcal P_Q$ to $\hat{\mathbf P}_{t,q}^{\rm mix}$ in Frobenius norm. Moreover, its full-CSI reference capture-power loss satisfies
\begin{equation}
	0 \le \operatorname{tr}\!\left[ \left( \mathbf P_{t,q}^{\rm ref,Q} - \hat{\mathbf P}_{t,q}^{Q} \right) \widetilde{\mathbf R}_{t,q}^{\rm UE} \right] \le 2Q \left\| \hat{\mathbf P}_{t,q}^{\rm mix} - \widetilde{\mathbf R}_{t,q}^{\rm UE} \right\|_2 .
	\label{eq:reference_capture_bound}
\end{equation}
\end{lemma}
\begin{IEEEproof}
	See Appendix~\ref{app:lem_near_reference_rank_q}.
\end{IEEEproof}

\textbf{Lemma~\ref{lem:near_reference_rank_q}} connects the covariance-estimation accuracy of SiFo to the quality of the final CSI subspace decision. The eigenspace extraction in \eqref{eq:final_projector} is optimal for the fused covariance estimate, and its loss relative to the UE full-CSI reference subspace is controlled by the deviation between $\hat{\mathbf P}_{t,q}^{\rm mix}$ and $\widetilde{\mathbf R}_{t,q}^{\rm UE}$. Consequently, calibration-neighbor averaging and confidence-weighted fusion affect the final subspace through the covariance-estimation error in \eqref{eq:reference_capture_bound}.
\vspace{-0.3cm}
\section{Practical Implementation of SiFo}
\label{sec:backbone_deployment}
The preceding sections establish the calibration principle and the projector-memory acquisition rule. This section specifies the deployment protocol of SiFo, separating offline preparation from online CSI subspace acquisition.

\subsection{Offline Pretraining and Calibration Memory Construction}
The offline preparation consists of two components: cross-site pretraining of the SST2 backbone and target-site construction of projector-labeled calibration memory.
\\ \indent
\textbf{Pretrained SST2 backbone:}
The backbone follows the unified site-specific feedback framework of~\cite{Zhao2026UnifiedLF}. It contains a learned SSB probing codebook and a subspace inference network trained jointly through the capture-efficiency objective. Concretely,
\\
\emph{Learned SSB probing codebook:}
$\mathbf B_{\boldsymbol{\Theta}}\in\mathbb C^{N_t\times K}$ is a learned beamforming matrix whose columns define the $K$ SSB probing beams. For UE $q$ with channel $\mathbf h_{s,q}$, the SSB stage produces the RSRP fingerprint $\mathbf r_{s,q}(\mathbf B_{\boldsymbol{\Theta}})$ following~\eqref{RSRP}. This module implements the sensing map $\mathbf h_{s,q}\mapsto \mathbf r_{s,q}(\mathbf B_{\boldsymbol{\Theta}})$ under the prescribed SSB probing budget.
\\
\emph{Subspace inference network:}
$H_{\boldsymbol{\Theta}}:\mathbb R^{K\times 1}\to\mathbb C^{N_t\times Q}$ is a learned inference network whose output defines the $Q$-dimensional CSI representation subspace. For UE $q$ with RSRP fingerprint $\mathbf r_{s,q}(\mathbf B_{\boldsymbol{\Theta}})$, the BS produces $ \hat{\mathbf U}_{s,q} = H_{\boldsymbol{\Theta}}\!\left( \mathbf r_{s,q}(\mathbf B_{\boldsymbol{\Theta}}) \right)$, and the corresponding subspace estimate is $\hat{\mathcal U}_{s,q}=\operatorname{span}(\hat{\mathbf U}_{s,q})$. The associated orthogonal projector is
\begin{equation}
	\tilde{\mathbf P}_{s,q}(\mathbf r_{s,q},\boldsymbol{\Theta}) = \hat{\mathbf U}_{s,q} \left( \hat{\mathbf U}_{s,q}^{H}\hat{\mathbf U}_{s,q} \right)^{-1} \hat{\mathbf U}_{s,q}^{H}.
	\label{eq:parametric_projector_implementation}
\end{equation}
In implementation, $H_{\boldsymbol{\Theta}}$ is a multi-layer perceptron (MLP).
\\ \indent
\textbf{Target-site calibration memory:}
After cross-site pretraining, the pretrained probing codebook $\mathbf B_{\boldsymbol{\Theta}_0^\star}$ is kept fixed for both target-site calibration and online serving. This fixed probing codebook ensures that calibration fingerprints and online query fingerprints are expressed in the same learned-beam coordinates.
\\ \indent
For each calibration UE $i\in\mathcal C_t$ in target site $t$, the BS obtains its normalized RSRP key $\tilde{\mathbf r}_{t,i}$ under $\mathbf B_{\boldsymbol{\Theta}_0^\star}$ according to \eqref{eq:calibration_key}. The full-CSI channel of the same UE is used offline to construct the direction projector $\mathbf P_{t,i}^{\star}$ in \eqref{eq:full_csi_direction_projector}. The two quantities are paired to form the calibration memory $\mathcal Y_t = \left\{ \left( \tilde{\mathbf r}_{t,i}, \mathbf P_{t,i}^{\star} \right) \right\}_{i\in\mathcal C_t}$, where $|\mathcal C_t|$ is the number of target-site calibration UEs.
\vspace{-0.3cm}
\subsection{Online CSI Subspace Acquisition}
For a served UE $q$ at target site $t$, the BS first obtains the RSRP fingerprint $\mathbf r_{t,q}$ and corresponding normalized fingerprint $\tilde{\mathbf r}_{t,q}$ under $\mathbf B_{\boldsymbol{\Theta}_0^\star}$. SiFo then identifies calibrated UEs whose learned-beam response shapes are closest to that of UE $q$. Specifically, the cosine similarity between UE $q$ and calibration UE $i$ is
\begin{equation}
	c_{q,i} = \tilde{\mathbf r}_{t,q}^{T}\tilde{\mathbf r}_{t,i}, \qquad i\in\mathcal C_t .
	\label{eq:online_cosine_similarity}
\end{equation}
For a given neighborhood size $m$, the nearest calibration set is
\begin{equation}
	\mathcal N_m(q) = \operatorname{Top}_m\!\left(\{c_{q,i}\}_{i\in\mathcal C_t}\right),
	\label{eq:nearest_calibration_set}
\end{equation}
where $\operatorname{Top}_m(\cdot)$ returns the indices of the $m$ largest cosine similarities. Their projector labels are averaged according to \eqref{ensemble} and \eqref{eq:multi_scale_memory}, producing the memory covariance estimate $\bar{\mathbf P}_{t,q}(\tilde{\mathbf r}_{t,q})$.
\\ \indent
The parametric branch uses the original SSB observation to generate $\hat{\mathbf U}_{t,q} = H_{\boldsymbol{\Theta}_t} \left( \mathbf r_{t,q}(\mathbf B_{\boldsymbol{\Theta}_t}) 	\right)$ and forms $\tilde{\mathbf P}_{t,q}(\mathbf r_{t,q},\boldsymbol{\Theta}_t)$ from $\hat{\mathbf U}_{t,q}$ as in \eqref{eq:parametric_projector_implementation}. In SiFo, $\boldsymbol{\Theta}_t=\boldsymbol{\Theta}_0^\star$. For the fine-tuning baseline, only the subspace inference network is updated while the probing codebook remains fixed so that calibration fingerprints do not need to be remeasured.
\\ \indent
The fusion rule in \eqref{eq:pmix} is governed by the memory coefficient $\alpha_q$, which determines the relative contribution of the cross-site parametric prior and the target-site projector-memory estimate. Since both terms aim to approximate the same local normalized covariance, the desired coefficient is determined by their relative estimation accuracy. To express the retrieval confidence, let $i_q^{(1)}$ denote the nearest calibrated UE in the normalized RSRP fingerprint space and define
\begin{equation}
	i_q^{(1)}
	=
	\arg\max_{i\in\mathcal C_t} c_{q,i},
	\qquad
	\kappa_q
	=
	\operatorname{clip}\!\left(c_{q,i_q^{(1)}},0,1\right).
	\label{eq:nearest_neighbor_similarity}
\end{equation}
The following proposition formalizes the ideal fusion coefficient under an mean square error (MSE) criterion.

\begin{proposition}[\emph{Optimal fusion weight}]\label{prop:monotone}\normalfont
Let $\tilde{\mathbf P}_{t,q}(\mathbf r_{t,q},\boldsymbol{\Theta}_t)$ and $\bar{\mathbf P}_{t,q}(\tilde{\mathbf r}_{t,q})$ be two estimators of the target local normalized covariance $\widetilde{\mathbf R}_{t,q}^{\rm UE}$. Let their Frobenius MSE be $\sigma_{\rm par}^2$ and $\sigma_{\rm mem}^2(q)$, respectively, and assume the two errors are uncorrelated. For the linear fusion form in \eqref{eq:pmix}, the MSE-optimal memory coefficient is given by
\begin{equation}
	\alpha_q^\star = \frac{\sigma_{\rm par}^2}{\sigma_{\rm par}^2+\sigma_{\rm mem}^2(q)}.
	\label{eq:opt_alpha}
\end{equation}
If $\sigma_{\rm mem}^2(q)$ is non-increasing in the retrieval confidence $\kappa_q$, then $\alpha_q^\star$ is monotone non-decreasing in $\kappa_q$.
\end{proposition}
\begin{IEEEproof}
	See Appendix~\ref{app:prop_monotone}.
\end{IEEEproof}
The optimal coefficient in \eqref{eq:opt_alpha} depends on estimator error variances that are unavailable during online deployment. \textbf{Proposition~\ref{prop:monotone}} nevertheless identifies $\kappa_q$ as a measurable reliability variable for the memory branch, i.e., under local stationarity, a closer calibrated fingerprint indicates a more accurate local covariance estimate and therefore a larger memory weight. SiFo therefore uses the monotone rule
\begin{equation}\label{eq:alpha_ada}
	\alpha_q = \kappa_q,
\end{equation}
which implements the trend directly from the RSRP-domain quantity already computed for neighbor selection. The fused covariance estimate is then converted into the protocol-constrained rank-$Q$ subspace by \eqref{eq:final_projector}.

\subsection{Online and Offline Overhead}
The overhead reduction enabled by SiFo can be understood by considering where different Type-II subspace acquisition schemes place the acquisition burden. Conv-T2 \cite{3GPP38214} places most of this burden in the online procedure: although no site-specific training is required, the subspace decision relies on CSI-RS refinement, UE-side search, and feedback. SST2 \cite{Zhao2026UnifiedLF} shifts this burden from online UE-side searching and feedback to BS-side inference from $K$ SSB RSRP measurements, thereby substantially reducing CSI-RS transmission and feedback overhead. However, this shift is achieved by training a separate model from a large labeled dataset for each target site. The proposed SiFo keeps the same low-overhead online sensing-and-reporting procedure as SST2, while replacing from-scratch target-site training with reusable cross-site pretraining and a moderate-size projector-labeled calibration memory. \emph{The resulting progression moves the dominant burden from per-UE online acquisition, to per-site model training, and finally to reusable pretraining plus lightweight target-site calibration, reducing the total cost of deploying SST2 at a new site.} The effective-rate impact of this burden shift is evaluated in Section~\ref{subsec:effective_rate}.


\section{Numerical Results}\label{sec:results}
This section first specifies the simulation setup and evaluation protocol, and then reports the main performance validation of the proposed SiFo.

\subsection{Simulation Setup}
\begin{table}[t]
	\small
	\centering
	\caption{Simulation settings}
	\label{tab:sim_setup}
	\begin{tabular}{ccc}
		\hline
		\textbf{Parameter} & \textbf{Description} & \textbf{Value} \\ \hline
		$f_c$ & Carrier frequency & 3.5GHz \\
		$\rm BW$ & Band Width & 10 MHz \\
		$N_t$ & Number of BS antennas & 64 \\
		$K$ & SSB codebook size & 16 \\
		$d$ & Antenna spacing & $\lambda/2$ \\
		$Q$ & Subspace rank  & 4 \\
		$\sigma_{\rm sh}^2$ & Log-variance of shadowing & 1 dB \\
		$\mathcal S_n$ & Noise power spectrum density & -170 dBm/Hz \\
		$P_t$ & Transmit Power & 40dBm \\
		$\mathcal{S}_m$ & Neighborhood size set & \{5, 10, 20\} \\
		$\rho_r$ & Receiving SNR & 10 dB \\
		\hline
	\end{tabular}
\end{table}
Experiments are based on ten DeepMIMO city-scale ray-tracing scenarios \cite{Alkhateeb2019DeepMIMO}, including New York (NY), Los Angeles (LA), Chicago (CHI), Houston (HOU), Phoenix (PHX), Philadelphia (PHL), Miami (MIA), San Diego (SD), Dallas (DAL), and San Francisco (SF). Cross-site pretraining uses leave-one-city-out (LOCO) training: for each target city, the pretrained model is trained on the remaining nine cities and then evaluated on this target city. The main LOCO comparison is reported over all 10 cities, while subsequent ablation studies use SD, the most difficult one as shown in Table~\ref{tab:main_q4}, as the representative city. Each city contributes approximately $20{,}000$ UE realizations over a large outdoor service area in average. The main metric is capture efficiency $\eta=\|\mathbf{P}_{Q}^{(\psi)} \mathbf h\|_2^2 / \|\mathbf h\|_2^2$, computed from the dominant $Q$-dimensional eigenspace returned by acquisition rule $\psi$. Effective rate is also complemented to demonstrate the better overhead-performance trade-off achieved by the proposed framework. Key simulation parameters are listed in Table~\ref{tab:sim_setup}. Unless otherwise stated, simulations are conducted based on these default parameters. In this paper, the number of BS antennas, carrier frequency, bandwidth, antenna spacing, SSB probing budget, and subspace rank are kept identical across deployed sites, so that the evaluation focuses on demonstrating the performance of propagation-domain transfer and target-site calibration efficiency. \\
\indent \textbf{Training:} The model for backbone pretraining is identical to the setup in \cite{Zhao2026UnifiedLF}. Furthermore, pretraining uses the Adam stochastic gradient optimizer \cite{KingmaBA2015Adam} with learning rate $10^{-3}$, batch size 512, and $20{,}000$ steps. Fine-tuning uses the same optimizer with L2-SP regularization \cite{LiHoiem2018L2SP} coefficient $\lambda_{\rm sp}=10^{-3}$. Memory and fine-tuning samples are drawn from the same target-site pool, and test evaluation is performed on the remaining held-out UEs. \\
\indent \textbf{Baselines:} We compare against five reference families. \emph{Conv-T2} is a high-overhead conventional reference, implemented with either a DFT-64 angular basis or an oversampled DFT-256 dictionary with orthogonal matching pursuit (OMP) search at UE. \emph{SST2} is the single-site parametric baseline trained with target-site labeled data. \emph{PT-SST2} is the cross-site pretrained model deployed without any target-site adaptation. \emph{FT-SST2} fine-tunes PT-SST2 with labeled target-site samples and does not use projector memory at inference. \emph{SiFo} denotes the proposed gradient-free target-site specialization with projector-labeled calibration UEs and no target-site fine-tuning. The number attached with the method indicates the number of used samples.

\subsection{Overall CSI Subspace Acquisition Performance}

\begin{table*}[t]
\centering
\caption{LOCO CSI-capture efficiency summary. Underlined values mark the high-overhead Conv-T2 DFT-256 upper reference. Bold values mark the best learning-based result in each city.}
\label{tab:main_q4}
\renewcommand{\arraystretch}{1.1}
\setlength{\tabcolsep}{3.8pt}
\footnotesize
\resizebox{\textwidth}{!}{%
\begin{tabular}{lccccccccccc}
\toprule
\textbf{Scheme} & \textbf{NY} & \textbf{LA} & \textbf{CHI} & \textbf{HOU} & \textbf{PHX} & \textbf{PHL} & \textbf{MIA} & \textbf{SD} & \textbf{DAL} & \textbf{SF} & \textbf{Avg} \\
\midrule
Conv-T2-DFT-64 & 0.8874 & 0.9004 & 0.9314 & 0.8903 & 0.9287 & 0.8989 & 0.9238 & 0.8761 & 0.8935 & 0.8929 & 0.9024 \\
Conv-T2-DFT-256 & \underline{0.9656} & \underline{0.9633} & \underline{0.9843} & \underline{0.9688} & \underline{0.9850} & \underline{0.9647} & \underline{0.9848} & \underline{0.9717} & \underline{0.9675} & \underline{0.9594} & \underline{0.9715} \\
\midrule
SST2-5000 & 0.8670 & 0.8520 & 0.9593 & 0.8044 & 0.9082 & 0.9090 & 0.9353 & 0.7736 & 0.8555 & 0.8676 & 0.8732 \\
SST2-Full & 0.9135 & 0.9068 & 0.9546 & 0.8978 & 0.9486 & 0.9156 & \textbf{0.9641} & 0.8745 & 0.9154 & 0.9048 & 0.9196 \\
\midrule
SiFo-200 & 0.8647 & 0.8630 & 0.9556 & 0.8631 & 0.9389 & 0.9200 & 0.9357 & 0.7954 & 0.8660 & 0.8698 & 0.8872 \\
SiFo-500 & 0.8794 & 0.8787 & 0.9608 & 0.8706 & 0.9436 & 0.9268 & 0.9429 & 0.8203 & 0.8786 & 0.8866 & 0.8988 \\
SiFo-1000 & 0.8903 & 0.8912 & 0.9652 & 0.8789 & 0.9485 & 0.9316 & 0.9490 & 0.8424 & 0.8891 & 0.8988 & 0.9085 \\
SiFo-5000 & \textbf{0.9138} & \textbf{0.9213} & \textbf{0.9735} & \textbf{0.9031} & \textbf{0.9592} & \textbf{0.9408} & 0.9621 & \textbf{0.8847} & \textbf{0.9179} & \textbf{0.9191} & \textbf{0.9292} \\
\bottomrule
\end{tabular}%
}
\end{table*}

Table~\ref{tab:main_q4} summarizes the LOCO CSI-capture results. Target-site labeled budget $B$ indicates the samples available for both memory restoration and fine-tuning. Conv-T2 DFT-256 is included as a high-overhead capture reference, while SiFo uses only the $K=16$ RSRP measurements produced by SSB probing online and performs no target-site gradient update. The performance increases consistently with calibration-memory size: SiFo-200 already exceeds SST2-5000 on average ($0.8872$ vs.\ $0.8732$), SiFo-1000 surpasses the DFT-64 Conv-T2 reference ($0.9085$ vs.\ $0.9024$), and SiFo-5000 reaches $0.9292$, outperforming SST2-Full on average and in nine out of ten cities. These results show that projector-labeled calibration memory uses target-site samples more effectively than retraining a separate parametric RSRP-to-subspace model, while the remaining gap to DFT-256 comes with a much lower online sensing and search burden.

\subsection{Impact of Target-Site Calibration}\label{subsec:budget_ablation}

\begin{figure}[t]
	\centering
	\includegraphics[scale=0.6]{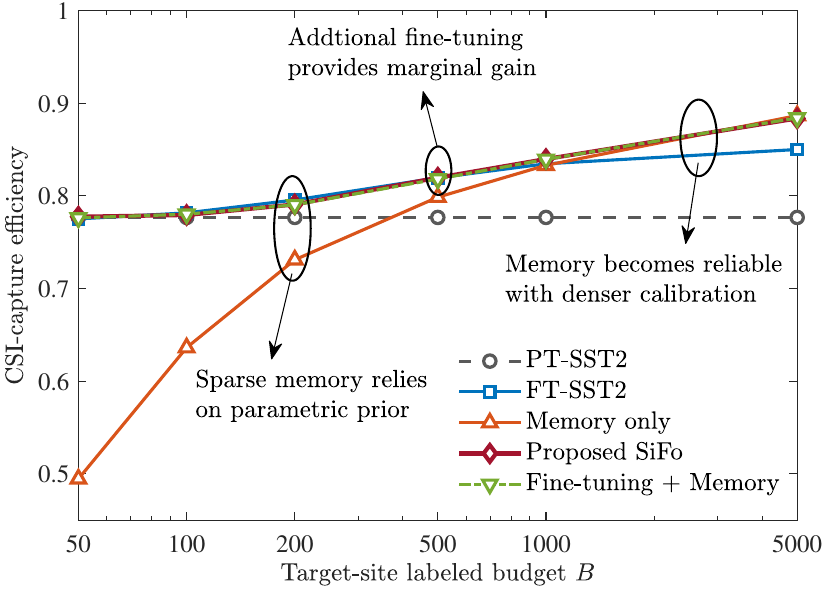}
	\caption{Comparison of different target-site calibration approaches.}
	\label{fig:budget_ablation}
\end{figure}

Fig.~\ref{fig:budget_ablation} verifies the complementarity between the parametric and memory branches. The pretrained model gives a stable but site-agnostic prior, while memory-only adaptation uses full-CSI-based projector evidence but becomes unreliable when calibrated neighbors are sparse. The proposed fusion falls back to the parametric prior at small budgets and increasingly exploits projector memory as calibration coverage improves, yielding both cold-start robustness and large-budget gain. Fine-tuning alone remains limited, and adding fine-tuning to memory brings only marginal improvement over the proposed gradient-free rule. This shows that SiFo uses target-site samples more data-efficiently by retaining them as local projector evidence rather than only updating model parameters.

\subsection{Impact of RSRP Sensing Coordinates}\label{subsec:key_ablation}

\begin{figure}[t]
	\centering
	\includegraphics[scale=0.6]{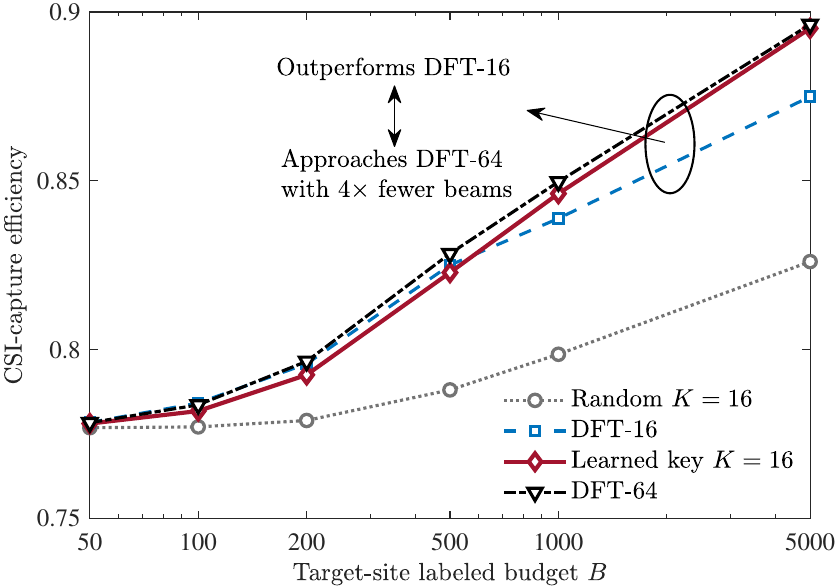}
	\caption{Comparison of different RSRP sensing coordinates.}
	\label{fig:key_ablation}
\end{figure}

Fig.~\ref{fig:key_ablation} isolates the role of the pretrained sensing coordinates. Random $K=16$ beams provide weak retrieval geometry, and the fixed DFT-16 grid is limited by its coarse angular resolution. The learned $K=16$ key consistently improves over DFT-16 at moderate and large calibration budgets, showing that cross-site pretraining shapes the RSRP fingerprint space for target-site neighbor search rather than merely learning a parametric predictor. It also approaches the DFT-64 reference while using four times fewer probing beams, supporting the low-overhead sensing principle in Remark~\ref{rem:learned_codebook_sensing_principle}.
\\ \indent
We further examine the mutual coherence of the learned beams. For a column-normalized codebook $\mathbf B=[\mathbf b_1,\ldots,\mathbf b_K]$, define the average off-diagonal Gram energy as
\begin{equation}
	\bar{\mu}_{\rm G}(\mathbf B) = \frac{1}{K(K-1)} \sum_{i\neq j} \left|\mathbf b_i^H\mathbf b_j\right|^2.
	\label{eq:avg_gram_energy}
\end{equation}
For the SD LOCO simulation, the learned $K=16$ probing codebook gives $\bar{\mu}_{\rm G}(\mathbf B_{\boldsymbol{\Theta}_0^\star})=0.0132$, indicating low average inter-beam correlation and no evident collapse to repeated beam measurements. This confirms the near-orthogonality aspect of the finite-budget probing principle discussed in Remark~\ref{rem:learned_codebook_sensing_principle}.

\subsection{Effective Spectral Efficiency Under Limited Target-Site Data}\label{subsec:effective_rate}

\begin{figure}[t]
	\centering
	\includegraphics[scale=0.6]{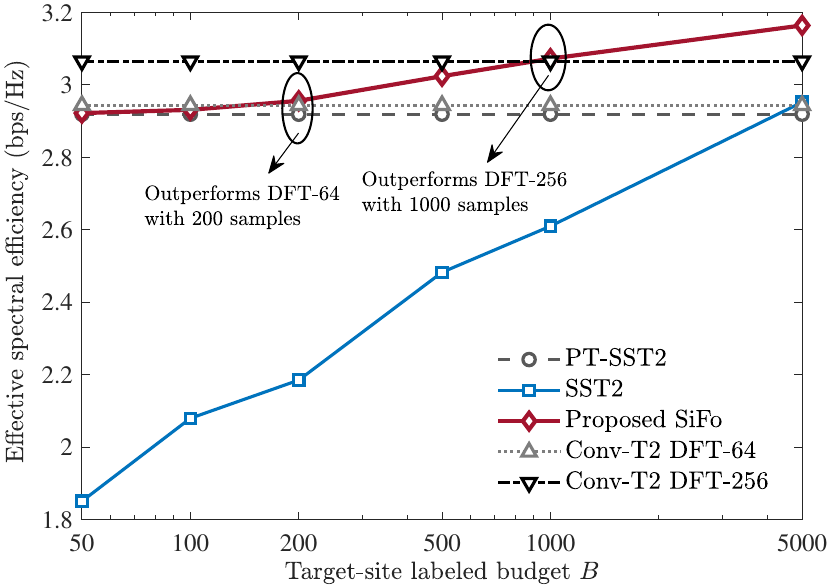}
	\caption{Effective spectral efficiency comparison under limited target-site data.}
	\label{fig:effective_rate}
\end{figure}

Fig.~\ref{fig:effective_rate} evaluates whether the CSI-capture gain translates into overhead-aware throughput. Conv-T2 is penalized by CSI-RS and feedback overhead, whereas SST2 and SiFo use the same low-overhead SSB probing and RSRP reporting procedure. Under this accounting, SST2 needs a larger target-site labeled set before surpassing the conventional Type-II references. SiFo crosses these references with fewer target-site samples because projector memory uses calibration data as local spatial evidence rather than only as training samples for a parametric mapping. This confirms that the data efficiency of SiFo also appears in effective spectral efficiency under practical online overhead.

\section{Conclusion}\label{sec:conclusion}
This paper developed SiFo, a wireless foundation model framework for low-overhead site-specific CSI feedback. The proposed framework combines cross-site pretraining with target-site projector memory. Concretely, the pretrained backbone provides a transferable SSB probing codebook and parametric subspace prior, while calibration memory supplies local full-CSI-based projector evidence for gradient-free target-site specialization. The analysis clarified why RSRP-only sensing cannot specify the phase-aware Type-II subspace, how rank-one full-CSI direction projectors serve as normalized covariance samples, and how neighbor projector averaging and confidence-weighted fusion recover a protocol-constrained rank-$Q$ subspace. Simulations over city-scale deployments showed that SiFo achieves strong CSI-capture performance with substantially lower online overhead than conventional Type-II feedback, uses target-site labeled samples more efficiently than single-site training or fine-tuning-only adaptation, and converts this gain into higher effective spectral efficiency under limited target-site data. These results indicate that foundation model pretraining and projector-labeled calibration memory provide a practical route toward data-efficient and low-overhead site-specific CSI feedback. Future work may extend the framework to online memory updates and broader cross-site adaptation settings.

\appendices
\section{Proof of \textbf{Theorem~\ref{thm:min_budget_isotropic_sensing}}} \label{app1}
The minimum collected sensing energy can be obtained via the Rayleigh-Ritz theorem as
\begin{equation}
	\min_{\|\mathbf h\|_2=1}\|\mathbf S^H\mathbf h\|_2^2 = \lambda_{\min}(\mathbf S\mathbf S^H).
\end{equation}
Since $J=N_t$ and all sensing beams have unit norm, $\operatorname{tr}(\mathbf S\mathbf S^H) = \operatorname{tr}(\mathbf S^H\mathbf S) = N_t$. Thus, the average eigenvalue of $\mathbf S\mathbf S^H$ is one, which implies $\lambda_{\min}(\mathbf S\mathbf S^H)\leq 1$. Hence, the optimal worst-case sensing energy is at most one.
\\\indent
Equality holds if and only if all eigenvalues of $\mathbf S\mathbf S^H$ are equal to one, because their average is one and their minimum is one. This is equivalent to $\mathbf S\mathbf S^H=\mathbf I_{N_t}$. Conversely, any feasible $\mathbf S$ satisfying this condition yields $\|\mathbf S^H\mathbf h\|_2^2 = \mathbf h^H\mathbf h = 1$ for every unit-norm $\mathbf h$, and therefore attains the upper bound.
\\\indent
Finally, the DFT-$N_t$ codebook is unitary, i.e., $\mathbf D_{N_t}\mathbf D_{N_t}^H=\mathbf I_{N_t}$. It therefore satisfies $\mathbf S\mathbf S^H=\mathbf I_{N_t}$. Its entries also have constant modulus, so it is compatible with the standard phase-only analog beam constraint. The proof is thus completed.

\section{Proof of \textbf{Lemma~\ref{lem:full_csi_direction_covariance}}}\label{app:lem_full_csi_direction_covariance}
Under path independence and the approximate orthogonality of steering vectors in \textbf{Assumption~\ref{assump:near_orthogonal_paths}}, the following results hold for each channel vector:
\begin{equation}
	\mathbf h_{t,i}\mathbf h_{t,i}^H
	\approx
	\sum_\ell
	|\alpha_{t,i,\ell}|^2
	\mathbf a(u_{t,i,\ell})\mathbf a^H(u_{t,i,\ell}).
\end{equation}
The corresponding channel power satisfies $\|\mathbf h_{t,i}\|_2^2\approx \sum_\ell |\alpha_{t,i,\ell}|^2$. Under the path-power concentration assumption, $\|\mathbf h_{t,i}\|_2^2\approx\operatorname{tr}(\mathbf R_{t,i}^{\rm UE})$ with high probability. Taking expectations and applying this concentration approximation to the denominator gives
\begin{equation}
	\mathbb E\!\left[
	\frac{\mathbf h_{t,i}\mathbf h_{t,i}^H}{\|\mathbf h_{t,i}\|_2^2}
	\right]
	\approx
	\frac{
	\sum_\ell
	\mathbb E[|\alpha_{t,i,\ell}|^2]
	\mathbf a(u_{t,i,\ell})\mathbf a^H(u_{t,i,\ell})
	}
	{\operatorname{tr}(\mathbf R_{t,i}^{\rm UE})}.
\end{equation}
Since the numerator is $\mathbf R_{t,i}^{\rm UE}$ under the same channel model, we have
\begin{equation}
	\mathbb E\!\left[
	\frac{\mathbf h_{t,i}\mathbf h_{t,i}^H}{\|\mathbf h_{t,i}\|_2^2}
	\right]
	\approx
	\frac{\mathbf R_{t,i}^{\rm UE}}
	{\operatorname{tr}(\mathbf R_{t,i}^{\rm UE})}
	=
	\widetilde{\mathbf R}_{t,i}^{\rm UE}.
\end{equation}
In the deterministic path-power case $|\alpha_{t,i,\ell}|^2\equiv\bar\sigma_{t,i,\ell}^2$, the above approximation becomes exact under the near-orthogonal angular model. The proof is thus completed.

\section{Proof of \textbf{Theorem~\ref{thm:memory_covariance_consistency}}}\label{app:thm_memory_covariance_consistency}
As $|\mathcal C_t|\to\infty$ with $m\to\infty$ and $m/|\mathcal C_t|\to0$, each nearest-neighbor neighborhood radius in the normalized RSRP fingerprint space shrinks to zero while the number of averaged neighbors grows. By \textbf{Assumption~\ref{assump:stationarity}}, all neighbors in the shrinking neighborhood share the same normalized UE-level covariance in the limit, i.e., $\widetilde{\mathbf R}_{t,i}^{\rm UE}\to\widetilde{\mathbf R}_{t,q}^{\rm UE}$. Applying \textbf{Lemma~\ref{lem:full_csi_direction_covariance}} to each full-CSI direction projector and then the law of large numbers to the growing average yields
\begin{equation}
	\mathbb E\!\left[ \hat{\mathbf P}_{t,q}^{(m)}(\tilde{\mathbf r}_{t,q}) \right] \xrightarrow{|\mathcal C_t|\to\infty} \widetilde{\mathbf R}_{t,q}^{\rm UE}.
\end{equation}
A finite average of several neighborhood estimators with the same limit has the same limit, which proves the multi-scale statement. The proof is thus completed.

\section{Proof of \textbf{Lemma~\ref{lem:near_reference_rank_q}}}\label{app:lem_near_reference_rank_q}
Let $\mathbf T_q=\widetilde{\mathbf R}_{t,q}^{\rm UE}$, $\mathbf A_q=\hat{\mathbf P}_{t,q}^{\rm mix}$, and $\mathbf E_q=\mathbf A_q-\mathbf T_q$. Since $\mathbf A_q$ is Hermitian, the Ky Fan maximum principle~\cite{HornJohnson2012MatrixAnalysis} gives
\begin{equation}
	\boldsymbol\Pi_Q(\mathbf A_q)
	\in
	\argmax_{\mathbf P\in\mathcal P_Q}
	\operatorname{tr}(\mathbf P\mathbf A_q),
\end{equation}
which proves \eqref{eq:rank_q_power_optimal}. For any $\mathbf P\in\mathcal P_Q$, $\mathbf P^2=\mathbf P$ and $\operatorname{tr}(\mathbf P)=Q$, so maximizing $\operatorname{tr}(\mathbf P\mathbf A_q)$ is equivalent to minimizing $\|\mathbf A_q-\mathbf P\|_F^2=\|\mathbf A_q\|_F^2+Q-2\operatorname{tr}(\mathbf P\mathbf A_q)$ over $\mathcal P_Q$.

Define
\begin{equation}
	\mathcal L_q
	=
	\operatorname{tr}\!\left[
	\left(
	\mathbf P_{t,q}^{\rm ref,Q}
	-
	\hat{\mathbf P}_{t,q}^{Q}
	\right)
	\mathbf T_q
	\right].
\end{equation}
Since $\mathbf P_{t,q}^{\rm ref,Q}$ maximizes $\operatorname{tr}(\mathbf P\mathbf T_q)$ over $\mathcal P_Q$, $\mathcal L_q\ge0$. Moreover, because $\hat{\mathbf P}_{t,q}^{Q}$ maximizes $\operatorname{tr}(\mathbf P\mathbf A_q)$ over $\mathcal P_Q$,
\begin{align}
	\mathcal L_q
	&=
	\operatorname{tr}\!\left[
	\left(
	\mathbf P_{t,q}^{\rm ref,Q}
	-
	\hat{\mathbf P}_{t,q}^{Q}
	\right)
	\mathbf A_q
	\right]
	-
	\operatorname{tr}\!\left[
	\left(
	\mathbf P_{t,q}^{\rm ref,Q}
	-
	\hat{\mathbf P}_{t,q}^{Q}
	\right)
	\mathbf E_q
	\right] \\
	&\le
	\left|
	\operatorname{tr}\!\left[
	\left(
	\mathbf P_{t,q}^{\rm ref,Q}
	-
	\hat{\mathbf P}_{t,q}^{Q}
	\right)
	\mathbf E_q
	\right]
	\right|.
\end{align}
By trace duality between the nuclear norm and the spectral norm,
\begin{equation}
	\mathcal L_q
	\le
	\left\|
	\mathbf P_{t,q}^{\rm ref,Q}
	-
	\hat{\mathbf P}_{t,q}^{Q}
	\right\|_*
	\|\mathbf E_q\|_2
	\le
	2Q\|\mathbf E_q\|_2,
\end{equation}
because both projectors have nuclear norm $Q$. The proof is thus completed.

\section{Proof of \textbf{Proposition~\ref{prop:monotone}}}\label{app:prop_monotone}
Under uncorrelated estimator errors, the cross-covariance term vanishes. The Frobenius MSE of $\hat{\mathbf P}_{t,q}^{\rm mix}=(1-\alpha)\tilde{\mathbf P}_{t,q}(\mathbf r_{t,q},\boldsymbol{\Theta}_t)+\alpha\bar{\mathbf P}_{t,q}(\tilde{\mathbf r}_{t,q})$ with respect to $\widetilde{\mathbf R}_{t,q}^{\rm UE}$ is $(1-\alpha)^2\sigma_{\rm par}^2+\alpha^2\sigma_{\rm mem}^2(q)$. This is a convex quadratic in $\alpha$. Setting its derivative to zero gives $\alpha_q^\star=\frac{\sigma_{\rm par}^2}{\sigma_{\rm par}^2+\sigma_{\rm mem}^2(q)}$.
If $\sigma_{\rm mem}^2(q)$ is non-increasing in $\kappa_q$, then $\alpha_q^\star$ is non-decreasing in $\kappa_q$ because $\alpha_q^\star$ is a decreasing function of $\sigma_{\rm mem}^2(q)$. The proof is thus completed.

\bibliographystyle{IEEEtran}
\bibliography{reference/mybib}

@techreport{3GPP38214,
	author      = {{3GPP}},
	title       = {{NR}; Physical Layer Procedures for Data},
	institution = {3rd Generation Partnership Project (3GPP)},
	type        = {Technical Specification},
	number      = {TS 38.214},
	year        = {2018}
}

@ARTICLE{Love2008LimitedFeedback,
	author={Love, David J. and Heath, Robert W. and N. Lau, Vincent K. and Gesbert, David and Rao, Bhaskar D. and Andrews, Matthew},
	journal={IEEE J. Sel. Areas Commun.}, 
	title={An overview of limited feedback in wireless communication systems}, 
	year={2008},
	month={Oct.},
	volume={26},
	number={8},
	pages={1341-1365},
	doi={10.1109/JSAC.2008.081002}}

@ARTICLE{Giordani2019BeamManagementNR,
	author={Giordani, Marco and Polese, Michele and Roy, Arnab and Castor, Douglas and Zorzi, Michele},
	journal={IEEE Commun. Surv. Tutorials}, 
	title={A Tutorial on Beam Management for {3GPP} {NR} at mmWave Frequencies}, 
	year={2019},
	volume={21},
	number={1},
	month={Sep.},
	pages={173-196},
	doi={10.1109/COMST.2018.2869411}}

@ARTICLE{Fu2023TutorialCodebooks,
	author={Fu, Xiaotian and Le Ruyet, Didier and Visoz, Raphael and Ramireddy, Venkatesh and Grossmann, Marcus and Landmann, Markus and Quiroga, Wilmar},
	journal={IEEE Access}, 
	title={A Tutorial on Downlink Precoder Selection Strategies for {3GPP} {MIMO} Codebooks}, 
	year={2023},
	volume={11},
	month={Dec.},
	pages={138897-138922},
	doi={10.1109/ACCESS.2023.3338866}}

@ARTICLE{Dreifuerst2024MLCodebook,
	author={Dreifuerst, Ryan M. and Heath, Robert W.},
	journal={IEEE Trans. Wireless Commun.}, 
	title={Machine Learning Codebook Design for Initial Access and {CSI} Type-{II} Feedback in Sub-6-{GHz} {5G NR}}, 
	year={2024},
	month={Jun.},
	volume={23},
	number={6},
	pages={6411-6424},
	doi={10.1109/TWC.2023.3331313}}

@ARTICLE{Dreifuerst2025NeuralCodebook,
	author={Dreifuerst, Ryan M. and Heath, Robert W.},
	journal={IEEE Trans. Wireless Commun.}, 
	title={Neural Codebook Design for {MIMO} Network Beam Management}, 
	year={2025},
	month={May},
	volume={24},
	number={5},
	pages={3909-3922},
	doi={10.1109/TWC.2025.3536290}}

@ARTICLE{Heng2022SiteSpecificProbing,
	author={Heng, Yuqiang and Mo, Jianhua and Andrews, Jeffrey G.},
	journal={IEEE Trans. Wireless Commun.}, 
	title={Learning Site-Specific Probing Beams for Fast mmWave Beam Alignment}, 
	year={2022},
	month={Jan.},
	volume={21},
	number={8},
	pages={5785-5800},
	doi={10.1109/TWC.2022.3143121}}

@INPROCEEDINGS{Ning2023RSRPCodebook,
	author={Ning, Xinzhi and Zhang, Shutao and Xue, Ye and Zheng, Xi and Shi, Qingjiang and Chang, Tsung-Hui},
	booktitle={Proc. {IEEE} Int. Workshop Signal Process. Adv. Wireless Commun. (SPAWC)}, 
	title={Learning Beams Adaptive to the Environment: An {RSRP}-based Codebook Design}, 
	year={2023},
	volume={},
	number={},
	pages={521-525},
	doi={10.1109/SPAWC53906.2023.10304486}}

@ARTICLE{Heng2024GridFree,
	author={Heng, Yuqiang and Andrews, Jeffrey G.},
	journal={IEEE Trans. Wireless Commun.}, 
	title={Grid-Free {MIMO} Beam Alignment Through Site-Specific Deep Learning}, 
	year={2024},
	month={Feb.},
	volume={23},
	number={2},
	pages={908-921},
	doi={10.1109/TWC.2023.3283475}}

@ARTICLE{Wu2024CKMBeamforming,
	author={Wu, Di and Zeng, Yong and Jin, Shi and Zhang, Rui},
	journal={IEEE Trans. Wireless Commun.}, 
	title={Environment-Aware Hybrid Beamforming by Leveraging Channel Knowledge Map}, 
	year={2024},
	month={Oct.},
	volume={23},
	number={5},
	pages={4990-5005},
	doi={10.1109/TWC.2023.3323941}}

@ARTICLE{Zeng2024CKMTutorial,
	author={Zeng, Yong and Chen, Junting and Xu, Jie and Wu, Di and Xu, Xiaoli and Jin, Shi and Gao, Xiqi and Gesbert, David and Cui, Shuguang and Zhang, Rui},
	journal={IEEE Commun. Surv. Tutorials}, 
	title={A Tutorial on Environment-Aware Communications via Channel Knowledge Map for 6G}, 
	year={2024},
	month={Feb},
	volume={26},
	number={3},
	pages={1478-1519},
	doi={10.1109/COMST.2024.3364508}}

@article{SSBFMAG2,
	title={Generative Site-Specific Beamforming for Next-Generation Spatial Intelligence},
	author={Wang, Zhaolin and Zhou, Zihao and Zhao, Cheng-Jie and Liu, Yuanwei},
	journal={arXiv preprint arXiv:2601.02301},
	year={2026}
}

@article{sim,
	title={Generative Site-Specific Beamforming via Information-Maximizing Codebook},
	author={Zhao, Cheng-Jie and Wang, Zhaolin and Liu, Yuanwei},
	journal={arXiv preprint arXiv:2602.12552},
	year={2026}
}

@ARTICLE{Ayach2014SpatiallySparse,
	author={Ayach, Omar El and Rajagopal, Sridhar and Abu-Surra, Shadi and Pi, Zhouyue and Heath, Robert W.},
	journal={IEEE Trans. Wireless Commun.},
	title={Spatially Sparse Precoding in Millimeter Wave {MIMO} Systems},
	year={2014},
	month={Mar.},
	volume={13},
	number={3},
	pages={1499-1513},
	doi={10.1109/TWC.2014.011714.130846}
}

@inproceedings{Ngo2014Favorable,
	author    = {Hien Quoc Ngo and Erik G. Larsson and Thomas L. Marzetta},
	title     = {Aspects of Favorable Propagation in Massive {MIMO}},
	booktitle = {Proc. 22nd Eur. Signal Process. Conf. (EUSIPCO)},
	pages     = {76--80},
	year      = {2014}
}

@techreport{3GPP38211,
	author      = {{3GPP}},
	title       = {{NR}; Physical Channels and Modulation},
	institution = {3rd Generation Partnership Project (3GPP)},
	type        = {Technical Specification},
	number      = {TS 38.211},
	year        = {2018}
}

@techreport{3GPP38215,
	author      = {{3GPP}},
	title       = {{NR}; Physical Layer Measurements},
	institution = {3rd Generation Partnership Project (3GPP)},
	type        = {Technical Specification},
	number      = {TS 38.215},
	year        = {2018}
}

@ARTICLE{Larsson2014MassiveMIMO,
	author={Larsson, Erik G. and Edfors, Ove and Tufvesson, Fredrik and Marzetta, Thomas L.},
	journal={IEEE Commun. Mag.},
	title={Massive {MIMO} for Next Generation Wireless Systems},
	year={2014},
	month={Feb.},
	volume={52},
	number={2},
	pages={186-195},
	doi={10.1109/MCOM.2014.6736761}
}

@INPROCEEDINGS{Alkhateeb2019DeepMIMO,
	author={Alkhateeb, Ahmed},
	booktitle={Proc. Inf. Theory Appl. Workshop (ITA)},
	title={{DeepMIMO}: A Generic Deep Learning Dataset for Millimeter Wave and Massive {MIMO} Applications},
	year={2019},
	month={Feb.},
	pages={1-8},
	address={San Diego, CA}
}

@article{Bommasani2021Foundation,
  author  = {Bommasani, Rishi and others},
  title   = {On the Opportunities and Risks of Foundation Models},
  journal = {arXiv preprint arXiv:2108.07258},
  year    = {2021}
}

@ARTICLE{WirelessLLM2024,
  author={Shao, Jiawei and Tong, Jingwen and Wu, Qiong and Guo, Wei and Li, Zijian and Lin, Zehong and Zhang, Jun},
  journal={J. Commun. Inf. Netw.},
  title={{WirelessLLM}: Empowering Large Language Models Towards Wireless Intelligence},
  year={2024},
  month={Jun.},
  volume={9},
  number={2},
  pages={99-112},
  doi={10.23919/JCIN.2024.10582827}
}

@ARTICLE{WiFo2024,
  author={Liu, Boxun and Gao, Shijian and Liu, Xuanyu and Cheng, Xiang and Yang, Liuqing},
  journal={Sci. China Inf. Sci.},
  title={{WiFo}: Wireless Foundation Model for Channel Prediction},
  year={2025},
  month={May},
  volume={68},
  number={6},
  pages={162302},
  doi={10.1007/s11432-025-4349-0}
}

@ARTICLE{Liu2026WiFoCF,
  author={Liu, Xuanyu and Gao, Shijian and Liu, Boxun and Cheng, Xiang and Yang, Liuqing},
  journal={IEEE Trans. Wireless Commun.},
  title={{WiFo-CF}: Wireless Foundation Model for {CSI} Feedback},
  year={2026},
  volume={25},
  pages={15039-15053},
  doi={10.1109/TWC.2026.3681311}
}

@article{Sheng2025WirelessFM,
  title={A Wireless Foundation Model for Multi-Task Prediction},
  author={Sheng, Yucheng and Wang, Jiacheng and Zhou, Xingyu and Liang, Le and Ye, Hao and Jin, Shi and Li, Geoffrey Ye},
  journal={arXiv preprint arXiv:2507.05938},
  year={2025}
}

@article{Alikhani2024LWM,
  author={Alikhani, Sadjad and Charan, Gouranga and Alkhateeb, Ahmed},
  title={Large Wireless Model: Foundation Model for Wireless Channels},
  journal={arXiv preprint arXiv:2411.08872},
  year={2024}
}

@ARTICLE{Yang2025WirelessGPT,
  author={Yang, Tingting and Zhang, Ping and Zheng, Mengfan and Shi, Yuxuan and Jing, Liwen and Huang, Jianbo and Li, Nan},
  journal={IEEE Netw.},
  title={{WirelessGPT}: A Generative Pre-Trained Multi-Task Learning Framework for Wireless Communication and Sensing},
  year={2025},
  volume={39},
  number={5},
  pages={58-65},
  doi={10.1109/MNET.2025.3579496}
}

@ARTICLE{Aboulfotouh2025WavesFM,
  author={Aboulfotouh, Ahmed and Mohammed, Elsayed and Abou-Zeid, Hatem},
  journal={IEEE Open J. Commun. Soc.},
  title={{6G WavesFM}: A Multimodal Foundation Model for Sensing, Communication, and Localization},
  year={2025},
  volume={6},
  pages={6792-6807},
  doi={10.1109/OJCOMS.2025.3600616}
}

@ARTICLE{Zhou2026SpectrumFM,
  author={Zhou, Fuhui and Liu, Chunyu and Zhang, Hao and Wu, Wei and Wu, Qihui and Quek, Tony Q. S. and Chae, Chan-Byoung},
  journal={IEEE J. Sel. Areas Commun.},
  title={{SpectrumFM}: A Foundation Model for Intelligent Spectrum Management},
  year={2026},
  volume={44},
  number={2},
  pages={4471-4488},
  doi={10.1109/JSAC.2025.3644783}
}

@ARTICLE{Mashaal2026IQFM,
  author={Mashaal, Omar Ahmad and Abou-Zeid, Hatem},
  journal={IEEE Open J. Commun. Soc.},
  title={{IQFM}: A Foundation Model for Wireless Multi-Antenna {IQ} Streams},
  year={2026},
  volume={7},
  pages={3483-3501},
  doi={10.1109/OJCOMS.2026.3661435}
}

@unpublished{Zhao2026UnifiedLF,
  author={Zhao, Cheng-Jie and Wang, Zhaolin and Zhao, Zongyao and Liu, Yuanwei},
  title={Bridging Standardized Codebook and Site-Specific Beamforming: A Unified Limited-Feedback Framework},
  note={Manuscript},
  year={2026}
}

@inproceedings{LiHoiem2018L2SP,
  author    = {Li, Xuhong and Grandvalet, Yves and Davoine, Franck},
  title     = {Explicit Inductive Bias for Transfer Learning with Convolutional Networks},
  booktitle = {Proc. Int. Conf. Mach. Learn. (ICML)},
  year      = {2018},
  pages     = {2825--2834}
}

@inproceedings{KingmaBA2015Adam,
  author    = {Kingma, Diederik P. and Ba, Jimmy},
  title     = {Adam: A Method for Stochastic Optimization},
  booktitle = {Proc. Int. Conf. Learn. Represent. (ICLR)},
  year      = {2015}
}

@book{HornJohnson2012MatrixAnalysis,
  author    = {Horn, Roger A. and Johnson, Charles R.},
  title     = {Matrix Analysis},
  edition   = {2nd},
  publisher = {Cambridge University Press},
  year      = {2012}
}
\end{document}